\documentclass{llncs}
\usepackage{xspace}
\usepackage{color}
\usepackage[linesnumbered,noend,noline]{algorithm2e}
\SetAlFnt{\small}
\SetInd{0.5em}{0.5em}
\usepackage[leqno]{amsmath}
\usepackage{amssymb}
\usepackage{tikz}
\usepackage{array}
\usepackage{tabularx}
\usepackage{url}
\usepackage{wrapfig}
\usepackage{graphicx}
\usepackage{epsfig}
\usepackage[draft]{fixme}
\newcommand{\IGNORE}[1]{\fixme{IGNORED: #1}}
\renewcommand{\IGNORE}[1]{}
\usepackage{listings}
\renewcommand{\figurename}{Fig.}
\newif\ifcomments
\commentsfalse
\usepackage{url}

\newcommand{\Daniel}[1]{\ifcomments\textcolor{green}{\footnotesize DK: #1}\fi}
\newcommand{\Jade}[1]{\ifcomments\textcolor{blue}{\footnotesize JA: #1}\fi}

\usepackage{xspace}

\newcommand{\myprog}{Prog.~}
\newcommand{\myalg}{Alg.~}
\newcommand{\myfig}{Fig.~}
\newcommand{\mythm}{Thm.~}
\newcommand{\mylem}{Lem.~}

\newcommand{\mysec}{Sec.~}

\newcommand{\wrt}{w.r.t.\xspace}
\newcommand{\eg}{e.g.\xspace}
\newcommand{\ie}{i.e.\xspace}
\newcommand{\cf}{cf.\xspace}

\let\mo\mathord

\let\implies\Rightarrow
\let\cimplies\Leftarrow

\newcommand{\transc}[1]{\mathop{{#1}^{+}}}

  \newcommand{\dfnshort}[2]{{#1}\triangleq{#2}}
  
  \newcommand{\dfn}[2]{{#1}~\triangleq~{#2}}
  
  \newcommand{\re}{r}

  \let\setset\setbf

  \newcommand{\evts}{\setset{E}}

  \newcommand{\tr}{\operatorname{tid}}

  \newcommand{\loc}{\operatorname{addr}}

  \newcommand{\val}{\operatorname{val}}

  \newcommand{\VX}{\operatorname{valid}}

\newcommand{\stacklabel}[1]
{\stackrel{\smash{\scriptstyle\textnormal{#1}}}}

  \newcommand{\po}{\operatorname{\textsf{po}}}
  \newcommand{\pobr}{\operatorname{\textsf{po-br}}}

  \newcommand{\dep}{\textsf{dp}}

  \newcommand{\rf}{\textsf{rf}}

  \newcommand{\grf}{\operatorname{\textsf{grf}}}
  \newcommand{\grfe}{\operatorname{\textsf{grfe}}}

  \newcommand{\rfe}{\textsf{rfe}}

  \newcommand{\rfi}{\textsf{rfi}}

  \newcommand{\ws}{\textsf{ws}}

  \newcommand{\ews}{\textsf{wse}}

  \newcommand{\fr}{\textsf{fr}}
  
  \newcommand{\efr}{\textsf{fre}}

  \newcommand{\ab}{\textsf{ab}}
  
  \newcommand{\fenced}{\textsf{fence}}
  
  \newcommand{\poi}[1]{\textsf{po\ensuremath{_{#1}}}}
    
  \newcommand{\ppo}{\operatorname{{\textsf{ppo}}}}
  
  \newcommand{\pio}{\operatorname{\textsf{po-loc}}}

  \newcommand{\ghb}{\operatorname{\textsf{ghb}}}

  \newcommand{\wf}{\operatorname{wf}}

  \newcommand{\acyclic}{\operatorname{acyclic}}
  
  \newcommand{\uniproc}{\operatorname{uniproc}}
  
  \newcommand{\thin}{\operatorname{thin}}

  \newcommand{\consensus}{\operatorname{consensus}}

    \newcommand{\s}{\operatorname{s}}

  \newcommand{\WW}{\operatorname{WW}}

  \newcommand{\RW}{\operatorname{RW}}
  
  \newcommand{\RR}{\operatorname{RR}}

  \newcommand{\WR}{\operatorname{WR}}

\newcommand{\import}{\operatorname{import}}

\let\prog\textsf
\let\as\texttt
\let\ltest\textbf

\newcommand{\proc}[1]{\ensuremath{P_{#1}}}
\newlength{\fmtlength}

\newcommand{\lwsync}{\operatorname{{\sf lwsync}}}

\newcommand{\lb}[4]{\((#1)\)\,#2#3#4}

\newcommand{\Cands}{relations}

\newcommand{\llb}[3]{\makebox[.6cm]{#1}#2#3}
\newcommand{\safes}{\operatorname{safe}}

\newcommand{\wpath}{\operatorname{path}}

\newcommand{\ssa}{\ensuremath{\text{\textsf{ssa}}}\xspace}
\newcommand{\parord}{\ensuremath{\text{\textsf{pord}}}\xspace}
\newcommand{\chains}{\ensuremath{\text{\textsf{chains}}}\xspace}
\newcommand{\reads}{\ensuremath{\text{\small{\textsf{reads}}}}\xspace}
\newcommand{\writes}{\ensuremath{\text{\small{\textsf{writes}}}}\xspace}
\newcommand{\fences}{\ensuremath{\text{\textsf{fences}}}\xspace}
\newcommand{\prf}{\ensuremath{\text{\textsf{prf}}}\xspace}

\newcommand{\rfsome}{\ensuremath{\text{\textsf{rf\_some}}}\xspace}
\newcommand{\aes}{{\sf ses}\xspace}
\newcommand{\clock}[1]{\ensuremath{\text{\textsf{clock}}_{#1}}}
\newcommand{\HB}{\operatorname{HB}}
\newcommand{\en}{\operatorname{en}}
\newcommand{\link}{\operatorname{link}}
\newcommand{\guard}{\operatorname{g}}
\newcommand{\prefix}{\operatorname{pref}}
\newcommand{\card}{\operatorname{card}}
\newcommand{\conc}{\operatorname{conc}}
\newcommand{\revdel}{\operatorname{rd}}
\newcommand{\nprop}{\textsf{prop}}
\newcommand{\nc}{\operatorname{nc}}
\newcommand{\ac}{\operatorname{ac}}
\newcommand{\bc}{\operatorname{bc}}
\newcommand{\sval}{\ensuremath{\text{\small{\textsf{val}}}}\xspace}
\newcommand{\Clocks}{\operatorname{\textsf{C}}}

\newcommand{\Vals}{\operatorname{\textsf{V}}}
\newcommand{\WRv}{\operatorname{\textsf{WR}}}

\newcommand{\insta}{\operatorname{\textsf{inst}}}

\newcommand{\notrelax}{\operatorname{\textsf{not\_relax}}}

\newcolumntype{Y}{@{}r@{\,}X}
\newcommand{\instab}[2]{\ \(#2\)  & \as{#1}}

\newcommand{\pset}[2]{\(\as{#1} \leftarrow \as{#2}\)}
\newcommand{\pstore}[2]{\pset{#2}{#1}}
\newcommand{\pload}[2]{\pset{#1}{#2}}

\newcommand{\plwarx}[3]{\as{lwarx #1,#2,#3}}
\newcommand{\pstwcx}[3]{\as{stwcx. #1,#2,#3}}
\newcommand{\pstw}[3]{\as{stw #1,#2,#3}}
\newcommand{\pbne}[1]{\as{bne #1}}
\newcommand{\pcmp}[2]{\as{cmpw #1,#2}}
\newcommand{\haut}{\rule{0ex}{2ex}}
\newcommand{\bas}{\rule[-1ex]{0.5ex}{0ex}}

\newcommand{\Rmw}[1][40ex]
{\begin{tabular}{rl}
\haut \instab{\as{loop:}}{} \\
\instab{\plwarx{r1}{0}{r5}}{(a_1)} \\
\instab{[\dots]}{}\\
\instab{\pstwcx{r2}{0}{r5}}{(a_2)} \\
\bas\instab{\pbne{loop}}{(b)} \\ 
\end{tabular}}
\newcommand{\Atom}[1][40ex]
{\begin{tabular}{rl}
\haut \instab{\plwarx{r1}{0}{r5}}{(a_1)} \\
\instab{[\dots]}{}\\
\bas\instab{\pstwcx{r2}{0}{r5}}{(a_2)} \\
\end{tabular}}
\newcommand{\Lo}[1][40ex]
{\begin{tabular}{rl}
\haut\instab{\as{loop:}}{} \\
\instab{\plwarx{r6}{0}{r3}}{(a_1)} \\
\instab{\pcmp{r4}{r6}}{(b)} \\
\instab{\pbne{loop}}{(c)} \\
\instab{\pstwcx{r5}{0}{r3}}{(a_2)} \\
\instab{\pbne{loop}}{(d)} \\
\instab{\as{isync}}{(e)} \\
\bas\instab{[\dots]}{}\\ 
\end{tabular}}
\newcommand{\ULo}[1][40ex]
{\begin{tabular}{rl}
\instab{[\dots]}{} \\
\haut\instab{\as{lwsync}}{(f)} \\ 
\bas\instab{\pstw{r4}{0}{r3}}{(g)} \\
\end{tabular}}
\newcommand{\Iriw}[1][60ex]
{\begin{tabularx}{#1}{Y|Y|Y|Y}
\multicolumn{2}{c|}{\haut\proc{0}} &
\multicolumn{2}{c|}{\proc{1}} &
\multicolumn{2}{c|}{\proc{2}} &
\multicolumn{2}{c}{\proc{3}} \\ \hline
\instab{\pload{r1}{x}}{(a)}\haut &
\instab{\pload{r3}{y}}{(c)} &
\instab{\pstore{1}{x}}{(e)} &
\instab{\pstore{1}{y}}{(f)}\\
\bas\instab{\pload{r2}{y}}{(b)} &
\instab{\pload{r4}{x}}{(d)} &
\instab{}{} &
\instab{}{}
\\ \hline
\multicolumn{8}{l}{\haut{}Allowed? \as{r1=1}; \as{r2=0}; \as{r3=1};
\as{r4=0};}
\end{tabularx}}
\newcommand{\AOE}[1][40ex]
{\begin{tabularx}{#1}{Y|Y}
\multicolumn{2}{c|}{\haut\proc{0}} &
\multicolumn{2}{c}{\haut\proc{1}} \\ \hline
\haut\instab{\pstore{1}{x}}{(a)} & \instab{\pstore{1}{y}}{(c)} \\
\bas\instab{\pload{r1}{y}}{(b)} & \instab{\pload{r2}{x}}{(d)} \\ \hline 
\multicolumn{4}{l}{Allowed? \as{r1=0}; \as{r2=0} } \\
\end{tabularx}}

  \renewenvironment{thebibliography}[1]{%
    \begin{oldthebibliography}{#1}%
      \setlength{\parskip}{0ex}%
      \setlength{\itemsep}{0ex}%
      \small
  }%
  {%
    \end{oldthebibliography}%
  }

\mathchardef\ordinarycolon\mathcode`\:
\mathcode`\:=\string"8000
\begingroup \catcode`\:=\active
\gdef:{\mathrel{\mathop\ordinarycolon}}
\endgroup

\begin{document}
\makeatletter
\gdef\thelstlisting{\@arabic\c@lstlisting}
\makeatother



\title{Partial Orders for Efficient BMC of~Concurrent~Software}

\author{Jade Alglave\inst{1} \and Daniel Kroening\inst{2}
\and Michael Tautschnig\inst{2,3}} 
\institute{University College London \and University of Oxford \and Queen Mary,
University of London}

\pagestyle{plain}

\maketitle

\begin{abstract}

The vast number of interleavings that a concurrent program can have is
typically identified as the root cause of the difficulty of automatic
analysis of concurrent software.  Weak memory is generally believed to make
this problem even harder. 
We address both issues by modelling programs' executions with partial orders
rather than the interleaving semantics (SC).  We implemented a software
analysis tool based on these ideas.  It scales to programs of sufficient size
to achieve first-time formal verification of non-trivial concurrent systems
code over a wide range of models, including SC, Intel x86 and IBM Power.

\end{abstract}

%
%




\section{Introduction}
\label{sec:introduction}


Automatic analysis of concurrent programs is a practical challenge. Hardly any
of the very few existing tools for concurrency will verify a thousand lines of
code~\cite{dkw08}.  Most papers name the number of \emph{thread interleavings}
that a concurrent program can have as a reason for the difficulty.
This view presupposes an execution model, namely 
\emph{Sequential Consistency} (SC)~\cite{lam79}, 
where an execution is a
\emph{total order} 
(more precisely an interleaving) of the instructions from different threads.
The choice of SC as the execution model poses at least two problems.

First, the large number of interleavings modelling the executions of a program
makes their enumeration intractable. \emph{Context bounded}
methods~\cite{qr05,mq05,lr09,eqr11} (which are unsound in general) and
\emph{partial order reduction}~\cite{pel93,god96,fg05} can reduce the number of
interleavings to consider, but still suffer from limited scalability.
Second, modern multiprocessors (\eg, Intel x86 or IBM Power) serve as a
reminder that SC is an inappropriate model.  Indeed, the \emph{weak memory
models} implemented by these chips allow more behaviours than SC.

We address these two issues by using \emph{partial orders} to model executions,
following~\cite{pra86,win86,bf94,pp96}. 
We also aim at practical verification of concurrent
programs~\cite{cks05,cf11,eqr11}.
Rarely have these two communities met.  Notable exceptions
are~\cite{sw10,sw11}, forming with~\cite{bam07} the closest related work. We
show that the explicit use of partial orders generalises these works to
concurrency at large, from SC to weak memory, without affecting efficiency.  

Our method is as follows: we map a program to a formula consisting of two
parts.  The first conjunct describes the data and control flow for each
thread of the program; the second conjunct describes the concurrent
executions of these threads as partial orders.  We prove that for
any satisfying assignment of this formula there is a valid execution \wrt
our models; and conversely, any valid execution gives rise to a satisfying
assignment of the formula.

Thus, given an analysis for sequential programs (the per-thread conjunct), we
obtain an analysis for concurrent programs.  For programs with bounded loops,
we obtain a sound and complete model checking method.  Otherwise, if the
program has unbounded loops, we obtain an exhaustive analysis up to a given
bound on loop unrollings, \ie, a bounded model checking method.

To experiment with our approach, we implement a \emph{symbolic decision
procedure} answering reachability queries over concurrent C programs \wrt~a
given memory model.  We support a wide range of models, including SC, Intel x86
and IBM Power. To exercise our tool \wrt~weak memory, we verify $4500$~tests
used to validate formal models against IBM Power chips~\cite{ssa11,mms12}.  
Our tool is the first to handle the subtle \emph{store atomicity
relaxation}~\cite{ag96} specific to Power and ARM.  

\IGNORE{MICHAEL: back up the following
comment with experiments: On the flip side, the fewer constraints you have, the
harder it is for the tool to prove that a satisfying assignment does not exist
(when there is no bug) (potentially).  Could mention here that you have
experiments to back up this statement in both the cases where there is a bug,
and the cases where there isn't.  }
We show that mutual exclusion is not violated in a queue mechanism of the
Apache HTTP server software. 
We confirm a bug in the worker synchronisation mechanism in PostgreSQL,
and that adding two fences fixes the problem. We
verify that the Read-Copy-Update mechanism of the Linux kernel 
preserves data consistency of the object it is protecting. 
For all examples we perform the analysis for a wide range of memory
models, from SC to IBM Power via Intel x86.

We provide the sources of our tool, our experimental logs and our benchmarks at
\url{http://www.cprover.org/wpo}. 
\section{Related Work\label{rw}}

We start with models of concurrency, then review tools proving the absence of
bugs in concurrent software, organised by techniques.

\paragraph{Models of concurrency}
Formal methods traditionally build on Lamport's SC~\cite{lam79}.
A year earlier, Lamport defined \emph{happens-before models}~\cite{lam78}.
The happens-before order is the smallest partial order containing the program
order and the relation between a write, and a read from this write. 

These models seem well suited for analyses relative to synchronisation,
\eg,~\cite{eqt07,ffy08,kw10}, because the relations they define are
oblivious to the implementation of the idioms.  Despite happens-before being
a partial order, most of~\cite{lam78} explains how to linearise it.  Hence,
this line of work often relies on a notion of total orders.  Partial orders,
however, have been successfully applied in verification in the context of
Petri nets~\cite{m92}, which have been linked to software verification
in~\cite{kkw12} for programs with a small state space. 

We (and~\cite{bm08:cav,gg08,sw10,sw11}) reuse the \emph{clocks} of~\cite{lam78}
to build our orders.  Yet we do not aim at linearisation or a transitive
closure, as this leads to a polynomial overhead of redundant constraints.


Our work goes beyond the definition and simulation of memory
models~\cite{gys04,hr06,tvd10,ssa11,mms12}.  
Implementing an executable version of the memory models is an important step,
but we go further by studying the validity of systems code in C (as opposed to
assembly or toy languages) \wrt~both a given memory model and a property.

The style of the model influences the verification process.  Memory models
roughly fall into two classes: operational and axiomatic.  The
operational style models executions via interleavings, with transitions
accessing buffers or queues, in addition to the memory (as on SC).  Thus
this approach inherits the limitations of interleaving-based verification. 
For example,~\cite{abp11} (restricted to Sun Total Store Order, TSO) bounds
the number of context switches.


Other methods use operational specifications of TSO, Sun Partial Store Order
(PSO) and Relaxed Memory Order (RMO) to place fences in a
program~\cite{kvy10,kvy11,lnp12}.  Abdulla et al.~\cite{aac12} address this
problem on an operational TSO, for finite state transition systems instead of
programs.  The methods of~\cite{kvy10,kvy11} have, in the words
of~\cite{lnp12}, ``severely limited scalability''.  The dynamic technique
presented in~\cite{lnp12} scales to $771$ lines but does not aim to be sound:
the tool picks an invalid execution, repairs it, then iterates.


Axiomatic specifications categorise behaviours by constraining relations on
memory accesses. Several hardware vendors adopt this
style~\cite{sparc:94,alpha:02} of specification; we build on the axiomatic framework
of~\cite{ams12} (cf. \mysec\ref{sec:model}).  CheckFence~\cite{bam07} also uses
axiomatic specifications, but does not handle the store atomicity relaxation of
Power and ARM.

%
%
%
\Daniel{What's new in terms of theory? Differences to tools should go into
the experimental section, or are maybe motivational. POPL people don't care
so much about new tools; they care about new theory supported by a tool.}
\Jade{In fact this whole section should be read as motivational. We want to highlight that existing techniques don't do so well.}


\begin{wrapfigure}[15]{r}{0.42\columnwidth} 
\vspace*{-4em}
\begin{lstlisting}[caption={Fibonacci from \cite{b12}},label={prog:fib},
  numbers=none,morekeywords={start_thread,assert}]
#define N 5
int x=1, y=1;

void thr1() {                          
  for(int k=0; k<N; ++k)               
    x=x+y; }                           

void thr2() {
  for(int k=0; k<N; ++k)
    y=y+x; }

int main() {
  start_thread(thr1);   
  start_thread(thr2);
  assert(x<=144 && y<=144);
  return 0; }
\end{lstlisting}
\end{wrapfigure}
\paragraph{Running example\label{rw:running}}

Below we use~\myprog\ref{prog:fib} (from the TACAS Software Verification
Competition~\cite{b12}) as an illustration.  The shared
variables~\lstinline{x} and~\lstinline{y} can reach
the~$(2\text{\lstinline{N}})$-th Fibonacci number, depending on the
interleaving of \lstinline{thr1} and \lstinline{thr2}. \myprog\ref{prog:fib}
permits at least $\mathcal{O}(2^{6\text{N}})$ interleavings of \lstinline{thr1}
and \lstinline{thr2}. In each loop iteration, \lstinline{thr1}
reads~\lstinline{x} and then~\lstinline{y}, and then writes~\lstinline{x};
\lstinline{thr2} reads~\lstinline{y} and~\lstinline{x}, and then
writes~\lstinline{y}.  Each interleaving of these two writes yields a unique
sequence of shared memory states.
Swapping, \eg,  the read of~\lstinline{y} in \lstinline{thr2} with the write
of~\lstinline{x} in \lstinline{thr1} does not affect the memory states, but
swapping the accesses to the same address does.

\paragraph{Interleaving tools}
Traditionally, tools are based on interleavings, and do not consider weak
memory. By contrast, we handle weak memory by reasoning in terms of
partial orders.\Daniel{The sentence about us is too mysterious.}\Daniel{I am
not sure that the following survey adds much value; it's too high-level to
provide any new insights. You have to ask yourself what the reader who doesn't
already know this will get from the rest of the paper.}\Jade{Well, I'm a reader
who doesn't know this already, and what I got from it was an idea of what
exists, and why the existing techniques struggle with concurrency.}

\emph{Explicit-state model checking} performs a search over states of a
transition system. SPIN~\cite{h97}, VeriSoft~\cite{g97} and Java
PathFinder~\cite{hp00,jys12} implement this approach; they adopt various forms
of \emph{partial order reduction}~(POR) to cope with the number of
interleavings.

POR reduces soundly the number of interleavings to
study~\cite{pel93,god96,fg05} by observing that a partial order gives rise to a
class of interleavings~\cite{maz89}, then picking only one interleaving in each
class.  \myprog\ref{prog:fib} is an instance where the effect of POR is
limited. We noted in \mysec\ref{rw:running} that amongst the
$\mathcal{O}(2^{6\text{N}})$ interleavings permitted by \myprog\ref{prog:fib},
only the interleavings of the writes give rise to unique sequences of states.
Hence distinct interleavings of the threads representing the same interleavings
of the writes are candidates for reduction. POR reduces the number of
interleavings by at least $2^{2\text{N}}$, but $\mathcal{O}(2^{4\text{N}})$
interleavings remain.

Explicit-state methods may fail to cope with large state spaces, even in a
sequential setting. Symbolic encodings~\cite{bcm90} can help, but the state
space often needs further reduction using, \eg, \emph{bounded model checking}
(BMC)~\cite{bcc99} or \emph{predicate abstraction}~\cite{gs97}.
These techniques may again also be combined with POR.
ESBMC~\cite{cf11} implements BMC. An instance of \myprog\ref{prog:fib} has a
fixed~\lstinline{N}, \ie, bounded loops. Thus BMC with~\lstinline{N} as bound is
sound and complete for such an instance.  ESBMC verifies \myprog\ref{prog:fib}
for \lstinline{N = 10} within $30$\,mins (\cf
\mysec\ref{sec:experiments}, \myfig\ref{fig:all-tools}).
SatAbs~\cite{cks05} uses predicate abstraction in a CEGAR loop; it completes no
more than \lstinline{N = 3} in $30$~mins as it needs multiple predicates per
interleaving, resulting in many refinement iterations. 
Our approach easily scales to, \eg,
\lstinline{N = 50}, in less than $20$\,s, and more than \lstinline{N=300}
within $30$\,mins,
as we build only a polynomial number of constraints, at worst cubic in the
number of accesses to a given shared memory address.

\paragraph{Non-interleaving tools} Another line of tools is not based on
interleavings.  The existing approaches do not handle weak memory and are
either incomplete (\ie, fail to prove the absence of a bug) or unsound (\ie,
might miss a bug due to the assumptions they make). 


\emph{Thread-modular reasoning}~\cite{j83,ffq02,fqs02,fq03,gpr11} is sound, but
usually incomplete.  Each read presumes guarantees about the values provided
by the environment.  Empty guarantees amount to fully non-deterministic
values, thus this is a trivially sound approach.
Our translation of
\mysec\ref{sec:c-to-aes} corresponds to empty guarantees.  The constraints of
\mysec\ref{sec:aes-to-ces}, however, make our encoding complete.  
\Daniel{The
above is a bad mixture of related work and contribution. Expect people not to
read the related work section.}

In \myprog\ref{prog:fib}, if we guarantee \lstinline{x<=144 && y<=144},
the problem becomes trivial, but finding this guarantee automatically is
challenging.  Threader~\cite{gpr11} fails for \lstinline{N=1}
(\cf \mysec\ref{sec:experiments}, \myfig\ref{fig:all-tools}).

Context bounded methods fix an arbitrary bound on \emph{context switches}
\cite{qr05,mq05,lr09,eqr11}.  This supposes that most bugs happen with few
context switches. Our method does not make this restriction. Moreover, we
believe that there is no obvious generalisation of these works to weak memory,
other than instrumentation as~\cite{abp11} does for TSO, \ie, adding
information to a program so that its SC executions simulate its weak ones. We
used our tool in SC mode, and applied the instrumentation of~\cite{abp11} to
it. On average, the instrumentation is $9$~times more costly
(\cf~\mysec\ref{sec:experiments}, \myfig\ref{fig:all-tools}). 
   
In \myprog\ref{prog:fib}, we need at least \lstinline{N} context switches to
disprove the assertion \lstinline{assert(x<=143 && y<=143)} (or any upper bound
to \lstinline{x} and \lstinline{y} that is the $(2\text{\lstinline{N}})$-th
Fibonacci number minus~$1$).  The hypothesis of the approach (\ie, small
context bounds suffice to find a bug) does not apply here; Poirot fails for
\lstinline{N}$\geq 1$ (cf \mysec\ref{sec:experiments},
\myfig\ref{fig:all-tools}).

Our work relates the most to~\cite{bam07,gg08,sw10,sw11}; we
discuss~\cite{gg08} below and detail~\cite{bam07,sw10,sw11} in
\mysec\ref{sec:cf-s-w-comparison}.\IGNORE{MICHAEL: we don't compare to gg08,
so either we should, or we should rephrase this sentence} These works use
axiomatic specifications of SC to compose the distinct threads.\IGNORE{jade:
watch out, we have to understand gg08 a bit more; maybe their token stuff is a
key} CheckFence~\cite{bam07} models SC with total orders and transitive closure
constraints; \cite{sw10,sw11} use partial orders like us.  \cite{sw10,sw11}
note redundancies of their constraints, but do not explain them; our semantic
foundations (\mysec\ref{sec:model}) allow us both to explain their redundancies
and avoid them (\cf  \mysec\ref{sec:cf-s-w-comparison}).

The encodings of \cite{bam07,sw10,sw11} are
$\mathcal{O}(\text{\lstinline{N}}^3)$ for~$N$ shared memory accesses \emph{to
any address};
\cite{gg08} is quadratic, but in the
number of threads times the number of per-thread transitions, which may include
arbitrary many local accesses.
Our encoding is $\mathcal{O}(\text{\lstinline{M}}^3)$,
with \lstinline{M} the
maximal number of events for a \emph{single address}.
By contrast, the encodings of \cite{gg08,sw10,sw11}
quantify over all addresses. \myprog\ref{prog:fib} has two addresses only, but
the difference is already significant: $(6\lstinline{N})^3$ for
\cite{bam07,gg08,sw10,sw11} vs. $2\times(3\lstinline{M})^3$ in our case, \ie $1/4$ of
the constraints (\cf\mysec\ref{sec:experiments}, \myfig\ref{fig:facts} for
other case studies).

\section{Context: Axiomatic Memory Model\label{sec:model}}

We use the framework of~\cite{ams12}, which provably embraces several
\emph{architectures}: SC~\cite{lam79}, Sun TSO (\ie the x86
model~\cite{oss09}), PSO and RMO
, Alpha
, and a
fragment of Power
. We present this framework 
via \emph{litmus tests}, as shown in \myfig\ref{fig:sb}. 

\begin{figure}[!t]
\begin{tabular}{m{.5\linewidth}m{.5\linewidth}}
\centerline{\AOE[.6\linewidth]}
&
  \scalebox{0.9}{
  \begin{tikzpicture}[>=stealth,thin,inner sep=0pt,text centered,shape=rectangle]
  \useasboundingbox (-1cm,0.5cm) rectangle (2cm,-2.6cm);
  
    \begin{scope}[minimum height=0.5cm,minimum width=0.5cm,text width=1.0cm]
      \node (a)  at (0, 0)  {\lb aWx{1}};
      \node (b)  at (0, -2)  {\lb bRy{0}};
      \node (c)  at (2, 0)  {\lb cWy{1}};
      \node (d)  at (2, -2)  {\lb dRx{0}};
    \end{scope}
		
    \path[->] (a) edge [out=225,in=135] node [left=0.1cm] {po} (b);
    \path[->] (b) edge [out=15,in=195] node [pos=0.8, below right=0.07cm] {fr} (c);
    \path[->] (c) edge [out=-45,in=45] node [left=0.1cm] {po} (d);
    \path[->] (d) edge [out=165,in=-15] node [pos=0.8, below left=0.07cm] {fr} (a);
  \end{tikzpicture}}
\end{tabular}
\vspace*{-8mm}
\caption{\label{fig:sb} Store Buffering (\ltest{sb})}
\vspace*{-4mm}
\end{figure}
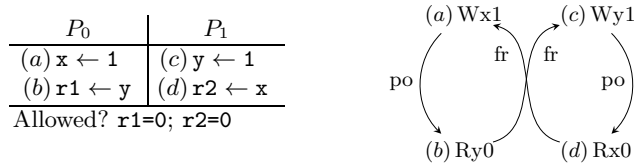

The keyword \emph{allowed} asks if the architecture permits the outcome ``{\tt
r1=1;r2=0;} {\tt r3=1;r4=0}''.  This relates to the event graphs of this program,
composed of relations over \emph{read and write memory events}.  A store
instruction (\eg \pstore{1}{x} on $P_0$) corresponds to a write event
(\llb{$(a)$}{W}{x}{$1$}), and a load (\eg \pload{r1}{y} on $P_0$) to a read
(\llb{$(b)$}{R}{y}{0}).  The validity of an execution boils down to the absence
of certain cycles in the event graph.  Indeed, an architecture allows an
execution when it represents a \emph{consensus} amongst the processors.  A
cycle in an event graph is a potential violation of this consensus.

If a graph has a cycle, we check if the architecture \emph{relaxes}
some relations. The consensus ignores relaxed relations,
hence becomes acyclic, \ie the architecture allows the final state. In
\myfig\ref{fig:sb}, on SC where nothing is relaxed, the cycle forbids
the execution. x86 relaxes the program order ($\po$ in \myfig\ref{fig:sb})
between writes and reads, thus a forbidding cycle no longer exists for 
$(a,b)$ and $(c,d)$ are relaxed.

\paragraph{Executions}
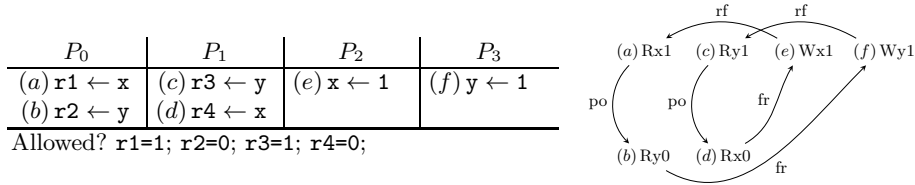
\begin{figure}[!t]
\begin{tabular}{m{.6\linewidth}@{}m{.4\linewidth}}
\centerline{\Iriw[\linewidth]}
&
  \scalebox{0.7}{
  \begin{tikzpicture}[>=stealth,thin,inner sep=0pt,text centered,shape=rectangle]
  \useasboundingbox (-1.5cm,1cm) rectangle (3.5cm,-3.2cm);
  
    \begin{scope}[minimum height=0.5cm,minimum width=0.5cm,text width=1.0cm]
      \node (a)  at (0, 0)  {\lb aRx{1}};
      \node (b)  at (0, -2)  {\lb bRy{0}};
      \node (c)  at (1.5, 0)  {\lb cRy{1}};
      \node (d)  at (1.5, -2)  {\lb dRx{0}};
      \node (e)  at (3.0, 0)  {\lb eWx{1}};
      \node (f)  at (4.5, 0)  {\lb fWy{1}};
    \end{scope}
		
    \path[->] (a) edge [out=225,in=135] node [left=0.1cm] {po} (b);
    \path[->] (c) edge [out=225,in=135] node [left=0.1cm] {po} (d);
    \path[->] (e) edge [out=150,in=30] node [above=0.1cm] {rf} (a);
    \path[->] (b) edge [out=-30,in=225,looseness=1.0] node [below right=0.07cm] {fr} (f);
    \path[->] (f) edge [out=150,in=30] node [above=0.1cm] {rf} (c);
    \path[->] (d) edge [out=30,in=-120] node [above left=0.07cm] {fr} (e);
  \end{tikzpicture}}
\end{tabular}
\vspace*{-8mm}
\caption{\label{fig:iriw} Independent Reads of Independent Writes (\ltest{iriw})}
\vspace*{-4mm}
\end{figure}
Formally, an \emph{event} is a read or a write memory access, composed of a
unique identifier, a direction R for read or W for write, a memory address,
and a value. We represent each instruction by the events it issues.
In~\myfig\ref{fig:iriw}, we associate the store \pstore{1}{$x$} on processor
$P_2$ to the event~\llb{$(e)$}{W}{$x$}{$1$}.   

We associate the program with an \emph{event structure} $\dfnshort{E}{(\evts,\po)}$,
composed of its events $\evts$ and the \emph{program order} $\po$, a
per-processor total order. We write $\dep$ for the relation (included in $\po$,
the source being a read) modelling \emph{dependencies} between instructions,
$\eg$ an \emph{address dependency} occurs when computing the address of a load
or store from the value of a preceding load. 
 
Then, we represent the \emph{communication} between processors leading to the
final state via an \emph{execution witness} $\dfnshort{X}{(\ws,\rf)}$, which
consists of two relations over the events.  First, the \emph{write
serialisation} $\ws$ is a per-address total order on writes which models the
\emph{memory coherence} widely assumed by modern
architectures
. It links a write $w$ to any write $w'$ to the same address that hits the
memory after $w$.  Second, the \emph{read-from} relation $\rf$ links a
write~$w$ to a read~$\re$ such that $\re$ reads the value written by~$w$. 
 
We include the writes in the consensus via the write serialisation.
Unfortunately, the read-from map does not give us enough information to embed
the reads as well. To that aim, we derive the \emph{from-read} relation $\fr$
from $\ws$ and $\rf$.  A read $r$ is in $\fr$ with a write $w$ when the write
$w'$ from which $r$ reads hit the memory before $w$ did.  Formally, we have:
$\dfn{(r, w) \in \fr}{\exists w', (w',r) \in \rf \wedge (w',w) \in \ws}$.
 
In \myfig\ref{fig:iriw}, the outcome corresponds to the execution on
the right if each memory location and register initially holds $0$. If {\tt
r1=1} in the end, the read $(a)$ read its value from the write $(e)$ on $P_2$,
hence $(e,a) \in \rf$. If {\tt r2=0}, the read $(b)$ read its value
from the initial state, thus before the write $(f)$ on $P_3$, hence $(b,f) \in
\fr$.  Similarly, we have $(f,c) \in \rf$ from {\tt r3=1}, and $(d,e) \in \fr$
from {\tt r4=0}. 

\paragraph{Relaxed or safe}

A processor can commit a write $w$ first to a store buffer, then to a cache,
and finally to memory. When a write hits the memory, all the processors agree
on its value. But when the write $w$ transits in store buffers and caches, a
processor can read its value through a read $r$ before the value is actually
available to all processors from the memory. In this case, the read-from
relation between the write $w$ and the read $r$ does not contribute to the
consensus, since the reading occurs in advance.  

We model this by some subrelation of the read-from $\rf$ being \emph{relaxed},
\ie not included in the consensus. When a processor can read from its own
store buffer~\cite{ag96} (the typical TSO/x86 scenario), we relax the internal
read-from $\rfi$. When two processors $P_0$ and $P_1$ can communicate privately
via a cache (a case of \emph{write atomicity} relaxation~\cite{ag96}),
we relax the external read-from $\rfe$, and call the corresponding write
\emph{non-atomic}. This is the main particularity of Power or ARM, and cannot
happen on TSO/x86.

Some program-order pairs are relaxed (\eg write-read pairs on x86)%
, \ie only a subset of $\po$ is guaranteed to occur in this order. 
%

When a relation is not relaxed, we call it \emph{safe}.
Architectures provide special \emph{fence} (or \emph{barrier}) instructions,
to prevent weak behaviours.  
Following~\cite{ams12}, the relation
${\fenced} \subseteq {\po}$ induced by a fence is \emph{non-cumulative} when it
orders certain pairs of events surrounding the fence, \ie $\fenced$ is safe.
The relation $\fenced$ is \emph{cumulative} when it makes writes atomic, \eg
by flushing caches. The relation $\fenced$ is \emph{A-cumulative}
(resp.~\emph{B-cumulative}) if $\rfe;\fenced$ (resp.~$\fenced;\rfe$) is
safe. 
When stores are atomic (\ie $\rfe$ is safe), \eg on 
TSO, we do not need cumulativity.

\paragraph{Architectures}
An \emph{architecture} $A$ determines the set $\safes_A$ of the \Cands~safe on
$A$, \ie the relations embedded in the consensus.  Following \cite{ams12}, we
consider the write serialisation $\ws$ and the from-read relation $\fr$ to be
always safe. 
SC relaxes nothing, \ie $\rf$ and $\po$ are safe.  TSO authorises the
reordering of write-read pairs and store buffering
(\ie $\poi{\textsf{WR}}$ and $\rfi$ are relaxed) but nothing else. 
We denote the safe subset of read-from, \ie the read-from relation globally
agreed on by all processors, by $\grf$.

Finally, an execution $(E,X)$ is \emph{valid} on $A$ when the three following
conditions hold.
1.~SC holds per address, \ie the communication 
and the program order for accesses with same address $\pio$ are compatible:
$\uniproc(E,X) \triangleq {\acyclic(\ws \cup \rf \cup \fr \cup \pio)}$. 
2.~Values do not come out of thin air, \ie there is no causal loop:
$\thin(E,X) \triangleq {\acyclic(\rf \cup \dep)}$.
3.~There is a consensus, \ie
the safe \Cands~do not form a cycle: $\consensus(E,X) \triangleq
{\acyclic({({\ws \cup \rf \cup \fr} \cup {\po})} \cap {\safes_A})}$.  
Formally: $\VX_{A}(E,X) \triangleq \uniproc(E,X) \wedge \thin(E,X) \wedge \consensus(E,X)$.

From the validity of executions we deduce a comparison of architectures: We say
that an architecture $A_2$ is \emph{stronger} than another one $A_1$ when the
executions valid on $A_2$ are valid on $A_1$. Equivalently we would say that
$A_1$ is \emph{weaker} than $A_2$. Thus, SC is stronger than any other
architecture discussed above.

\section{Symbolic event structures\label{sec:c-to-aes}}


For an architecture~$A$ and \emph{one} execution witness~$X$, the framework of
\mysec\ref{sec:model} determines if $X$ is valid on $A$. To prove reachability
of a program state, we need to reason about all its executions.  To do so
efficiently, we use symbolic representations capturing all possible executions
in a single constraint system.  We then apply SAT or SMT solvers to decide if a
valid execution exists for $A$, and, if so, get a satisfying assignment
corresponding to an execution witness.  

As said in~\mysec\ref{sec:introduction}, we build two conjuncts. The first one,
$\ssa$, represents the data and control flow per thread. The second, \parord,
captures the communications between threads (cf.~\mysec\ref{sec:aes-to-ces}).
We include a reachability property in~\ssa; the program has a valid execution
violating the property iff~$\ssa \wedge \parord$ is satisfiable.

We mostly use \emph{static single assignment form} (SSA)  of the input program to build \ssa
(\cf~\cite{cky03} for details).
In this SSA variant, each equation is augmented with a \emph{guard}: the guard
is the disjunction over all conjunctions of branching guards on paths to the
assignment.
To deal with concurrency, we use a fresh index for each occurrence of a given
shared memory variable, resulting in a fresh symbol in the formula.  
CheckFence~\cite{bam07} and~\cite{sw10,sw11} use a similarly modified encoding.

Together with \ssa, we build a \emph{symbolic event structure} (\aes).  As
detailed below, it captures basic program information needed to build the
second conjunct \parord in \mysec\ref{sec:aes-to-ces}.
\myfig\ref{ssa:iriw} illustrates this section: the formula \ssa on top
corresponds to the \aes beneath.

\begin{figure}
  \centering
\begin{align*}
&\text{\lstinline{main}} &  & P_0 & & P_1 & & P_2 & & P_3\\[-.5em]
& x_0 = 0 \\[-0.5em]
\wedge\; & y_0 = 0 \\[-1.3em]
         &   &  \wedge\; &r1_0^1 = x_1
                    & \wedge\;  &r3_0^2 = y_2
                         & \wedge\;  &x_3 = 1
                              & \wedge\;  & y_3 = 1 \\[-0.4em]
         &   &  \wedge\; &r2_0^1 = y_1
                    & \wedge\;  &r4_0^2 = x_2\\[-0.5em]
\wedge\; & \nprop
\end{align*}
\vspace*{-6mm}

\parbox{.96\linewidth}{\ \hrule height .000008mm \smallskip}
\tabcolsep0pt
\begin{tabular}{ll@{\extracolsep{1.2em}}l@{\extracolsep{0.2em}}l@{\extracolsep{1.2em}}l@{\extracolsep{0.2em}}l@{\extracolsep{1.2em}}l@{\extracolsep{0.2em}}l@{\extracolsep{1.2em}}l@{\extracolsep{0.2em}}l}
\llb{$(i_0)$}{W}{x}{$x_0$}  &   \\
\llb{$(i_1)$}{W}{y}{$y_0$}  &   \\[-.5em]
                          & &  \llb{$(a)$}{R}{x}{$x_{ 1}$}  &  
                               & \llb{$(c)$}{R}{y}{$y_{ 2}$}  & 
                                  & \llb{$(e)$}{W}{x}{$x_{ 3}$}  & 
                                       & \llb{$(f)$}{W}{y}{$y_{ 3}$}  &  \\
                          & &  \llb{$(b)$}{R}{y}{$y_{ 1}$}  &  
                               &  \llb{$(d)$}{R}{x}{$x_{ 2}$}  &  
\end{tabular}
\vspace*{-1mm}
\caption{\small{The formula~\ssa for~\ltest{iriw} (\myfig\ref{fig:iriw})
with~$\nprop=({r1}_{0}^{1}=1 \wedge {r2}_{0}^{1}=0 \wedge {r3}_{0}^{2}=1
\wedge {r4}_{0}^{2}=0)$, and its~\aes\label{ssa:iriw}} (guards
omitted since all true)}
\vspace*{3pt}
\end{figure}

\paragraph{Static single assignment form (SSA)}
To encode \ssa we use a variant of SSA~\cite{cfrwz91} and loop unrolling. 
The details of this encoding are in~\cite{cky03}, except for differences in the
handling of shared memory variables, as explained below.

In SSA, each occurrence of a program variable is annotated with an index.
We turn assignments in SSA form into equalities, with distinct indexes yielding
distinct symbols in the resulting equation.  For example, the assignment {\tt
x:=x+1} results in the equality $x_1 = x_0+1$.  We use unique indexes for
assignments in loops via loop unrolling: repeating {\tt x:=x+1} twice yields
$x_1 = x_0+1$ and $x_2 = x_1+1$.  Control flow join points yield additional
equations involving the guards of branches merging at this point (see
\cite{cky03} for details).

In concurrent programs, we also need to consider join points due to
communication between threads, \ie, \emph{concurrent SSA form}
(CSSA)~\cite{lmp97}.  
To deal with weaker models, we use a fresh index for each occurrence of a given
shared memory variable, resulting in a fresh symbol in the formula.  Thus, each
occurrence may take non-deterministic values, \ie this approach
over-approximates the behaviours of a program. If~{\tt x} is shared in the
above example, the modified SSA encoding of the second loop unrolling becomes
$x_3 = x_2+1$, breaking any causality between the first loop iteration (encoded
as $x_1 = x_0+1$) and the second one.  Sinha and Wang~\cite{sw10,sw11} use the
same approach, but since they consider SC only, their use of fresh indexes may
produce more symbols than necessary.

By adding the negation of the reachability property to be checked to our
(over-approximating) SSA equations, we obtain a formula~\ssa that is
satisfiable if there exists a (concurrent) counterexample violating the
property.  As this is an over-approximation, the converse need not be true,
\ie, a satisfying assignment of \ssa may constitute a spurious counterexample.
\mysec\ref{sec:aes-to-ces} restores precision using the~$\parord$ constraints
derived from the~\aes.

In \ssa, memory addresses map to unique symbols via the (symbolic) pointer
dereferencing of~\cite[\mysec{4}]{cky03}. In the weak memory case, we ensure
this by using analyses sound for this setting~\cite{akl11}.

The top of \myfig\ref{ssa:iriw} gives \ssa for \myfig\ref{fig:iriw}. We print a
column per thread, vertically following the control flow, but it forms a single
conjunction. Each occurrence of a program variable carries its SSA index as a
subscript. Each occurrence of the shared memory variables \lstinline{x} and
\lstinline{y} has a unique SSA index.  Here we omit the guards, as this program
does not use branching or loops.

\paragraph{From SSA to symbolic event structures}
A symbolic event structure (\aes)~$\dfnshort{\gamma}{(\mathbb{S}, \po)}$ is a
set~$\mathbb{S}$ of \emph{symbolic events} and a \emph{symbolic program
order}~$\po$.  A symbolic event holds a \emph{symbolic value} instead of a
concrete one as in \mysec\ref{sec:model}. We define $\guard(e)$ to be the
Boolean guard of a symbolic event~$e$, which corresponds to the guard of the
SSA equation as introduced above.
We use these guards to build the executions of \mysec\ref{sec:model}: 
a guard evaluates to true if the branch is taken, false otherwise.
The symbolic program order $\po(\gamma)$ gives a list of symbolic events per
thread of the program. The order of two events in~$\po(\gamma)$ gives the
program order in a concrete execution if both guards are true.

Note that $\po(\gamma)$ is an implementation-dependent
linearisation of the branching structure of a thread, induced by the path
merging applied while constructing the SSA form.  For instance, {\tt if $e_1$
then $e_2$ else $e_3$} could be linearised as either $(e_1, e_2, e_3)$ or $(e_1,
e_3, e_2)$ as any two events of a concrete execution ($e_1$ and $e_2$, or $e_1$
and $e_3$) remain in program order.
The original branching structure, \ie, the unlinearised symbolic program order induced
by the control flow graph, is maintained in the relation $\pobr(\gamma)$. For the above
example, $\pobr(\gamma)$ contains $(e_1, e_2)$ and $(e_1, e_3)$.
 
We build the \aes~$\gamma$ alongside the SSA form, as follows. Each occurrence
of a shared program variable on the right-hand side of an assignment becomes a
\emph{symbolic read}, with the SSA-indexed variable as symbolic value, and the
guard is taken from the SSA equation.
Similarly, each occurrence of a shared program variable on the left-hand side
becomes a \emph{symbolic write}.  Fences do not affect memory states in a
sequential setting, hence do not appear in SSA equations. We simply add a fence
event to the \aes when we see a fence.  We take the order of assignments per
thread as program order, and mark thread spawn points.

At the bottom of \myfig\ref{ssa:iriw}, we give the \aes of \ltest{iriw}. Each
column represents the symbolic program order, per thread. We use the same
notation as for the events of \mysec\ref{sec:model}, but values are SSA symbols.
Guards are omitted again, as 
they all are trivially true. We depict the thread spawn events by starting the
program order in the appropriate row. Note that we choose to put the two
initialisation writes in program order on the main thread.

\paragraph{From symbolic to concrete event structures\label{sec:concretisation}}

To relate to the models of \mysec\ref{sec:model}, we \emph{concretise} symbolic
events. A satisfying assignment to $\ssa \wedge \parord$, as computed by a SAT
or SMT solver, induces, for each symbolic event, a concrete value (if it is a
read or a write) and a valuation of its guard (for both accesses and fences).
A valuation $\Vals$ of the symbols of \ssa includes the values of each symbolic
event. Since guards are formulas that are part of \ssa, $\Vals$ allows us to
evaluate the guards as well. For a valuation $\Vals$, we write
$\conc(e_s,\Vals)$ for the concrete event corresponding to $e_s$, if there is
one, \ie, if $\guard(e_s)$ evaluates to true under~$\Vals$. 

The concretisation of a set $\mathbb{S}$ of symbolic events is a set
$\mathbb{E}$ of concrete events, as in \mysec\ref{sec:model}, s.t.~for
each $e \in \mathbb{E}$ there is a symbolic version $e_s$ in $\mathbb{S}$.
We write $\conc(\mathbb{S},\Vals)$ for this concrete set $\mathbb{E}$.
The concretisation $\conc(\textsf{r}_s,\Vals)$ of a symbolic relation
$\textsf{r}_s$ is the relation $\{(x,y) \mid \exists (x_s,y_s) \in \textsf{r}_s.
  x = \conc(x_s,\Vals) \wedge y = \conc(y_s,\Vals) \}$.

Given an \aes $\gamma$, $\conc(\gamma,\Vals)$ is the event structure 
(\cf \mysec\ref{sec:model}), whose set of events is the concretisation of the
events of~$\gamma$ \wrt~$\Vals$, and whose program order is the
concretisation of $\po(\gamma)$ \wrt~$\Vals$. 
For example, the graph of \myfig\ref{fig:iriw} (erasing the $\rf$ and $\fr$
relations) is a concretisation of the \aes of \ltest{iriw}
(\cf \myfig\ref{ssa:iriw}).

\section{Encoding the communication and weak memory relations symbolically\label{sec:aes-to-ces}}
For an architecture $A$ and an \aes $\gamma$, we need to represent the
communications (\ie, $\rf, \ws$ and $\fr$) and the weak memory relations (\ie,
$\ppo_{A}, \grf_{A}$ and $\ab_{A}$) of \mysec\ref{sec:model}.  We encode them
as a formula~$\parord$, s.t.~$\ssa \wedge \parord$ is satisfiable iff there is
an execution valid on~$A$ violating the property encoded in~$\ssa$.
We avoid transitive closures to obtain a small number of constraints.
We start with an informal overview of our approach, then describe how we encode
partial orders, and finally detail the encoding for each relation of
\mysec\ref{sec:model}.

\paragraph{Overview}
We present our approach on \ltest{iriw} (\myfig\ref{fig:iriw}) and its \aes
$\gamma$ (\myfig\ref{ssa:iriw}).  In \myfig\ref{fig:iriw}, we represent only
one possible execution, namely the one corresponding to the (non-SC) final
state of the test at the top of the figure.  In this section, we generate
constraints representing all the executions of \ltest{iriw} on a given
architecture. We give these constraints, for the address $x$ in
\myfig\ref{fig:constraints-iriw} in the SC case (for brevity we skip $y$,
analogous to $x$). Weakening the architecture removes some constraints: for
Power, we omit the (rf-grf) and (ppo) constraints.  For TSO, all constraints
are the same as for SC.

In \myfig\ref{fig:constraints-iriw}, each symbol $c_{a b}$ is a \emph{clock
constraint}, representing an ordering between the events $a$ and $b$. A
variable $s_{wr}$ represents a read-from between the write $w$ and the read
$r$. 

The constraints of \myfig\ref{fig:constraints-iriw} represent the preserved
program order (\cf \mysec\ref{sec:ppo}), \eg, on SC or TSO the read-read pairs
$(a,b)$ on $P_0$ (ppo $P_0$) and $(c,d)$ on $P_1$ (ppo $P_1$), but nothing on
Power. We generate constraints for the read-from (\cf \mysec\ref{pord:rf}), for
example (rf-some $x$); the first conjunct $s_{i_0 a} \vee s_{ea}$ concerns the
read $a$ on $P_0$. This means that~$a$ can read either from the initial write
$i_0$ or from the write $e$ on $P_2$.  The selected read-from pair also implies
equalities of the values written and read (rf-val~$x$): for instance,~$s_{i_0
a}$ implies that~$x_1$ equals the initialisation~$x_0$.  The
architecture-independent constraints for write serialisation (\cf
\mysec\ref{sec:ws}) and from-read (\cf \mysec\ref{sec:fr}) are specified as
(ws~$x$) and (fr~$x$); (ws~$y$) and (fr~$y$) are analogous. As there are no
fences in \ltest{iriw}, we do not generate any memory fence constraints (\cf
\mysec\ref{sec:ab}).

We represent the execution of \myfig\ref{fig:iriw} as follows.
For $(e,a)$ and $(i_0,d) \in \grf$, we have the constraint $s_{ea} \implies c_{ea}$
and $s_{i_0 d} \implies c_{i_0 d}$ in (rf-grf $x$). This means that $a$ reads
from $e$ (as witnessed by $s_{ea}$), and that we record that $e$ is ordered
before $a$ in $\grf$ (as witnessed by $c_{ea}$); \emph{idem} for $d$ and $i_0$.
To represent
$(d,e) \in \fr$, we pick the appropriate constraint in (fr $x$), namely $(s_{i_0
d} \wedge c_{i_0 e}) \implies c_{de}$. This reads ``if $d$ reads from $i_0$ and
$i_0$ is ordered before $e$ (in $\ws$, because $i_0$ and $e$ are two writes to
$x$), then $d$ is ordered before $e$ (in $\fr$).''  

Together with (ppo $P_0$) and (ppo $P_1$), these constraints represent the
execution in \myfig\ref{fig:iriw}. We cannot find a satisfying assignment of
these constraints, as this leads to both $a$ before $b$ (by (ppo $P_0$)) and
$b$ before $a$ (by (fr $y$), (rf-grf $y$), (ppo $P_1$), (fr $x$) and (grf
$x$)). On Power, however, we neither have the ppo nor the grf constraints, hence we
can find a satisfying assignment.

\begin{figure}[!t]
\begin{align}
  \tag{rf-val $x$} &(s_{i_0 a} \implies x_1 = x_0) \wedge (s_{i_0 d} \implies x_2 = x_0) \wedge\\[-0.5em]
  \notag &(s_{e a} \implies x_1 = x_3) \wedge (s_{e d} \implies x_2 = x_3) \\[-0.3em]
  \tag{rf-grf $x$} &(s_{i_0 a} \implies c_{i_0 a}) \wedge (s_{e a} \implies c_{e a}) \wedge\\[-0.5em]
  \notag &(s_{i_0 d} \implies c_{i_0 d}) \wedge (s_{e d} \implies c_{e d})\\[-0.3em]
  \tag{rf-some $x$} &(s_{i_0 a} \vee s_{e a}) \wedge (s_{i_0 d} \vee s_{e d})\\[-0.3em]
%
  \tag{ws $x$} &\neg c_{i_0 e} \Rightarrow c_{e i_0}\\
%
%
  \tag{fr $x$} &((s_{i_0 a} \wedge c_{i_0 e}) \implies c_{a e}) \wedge ((s_{i_0 d} \wedge c_{i_0 e}) \implies c_{d e}) \wedge\\[-0.5em]
  \notag &((s_{e a} \wedge c_{e i_0}) \implies c_{a i_0}) \wedge ((s_{e d} \wedge c_{e i_0}) \implies c_{d i_0})\\[-0.3em]
%
  \tag{ppo main} &c_{i_0 i_1}
    \qquad \text{(ppo $P_0$)\ \ } c_{a b}
    \qquad \text{(ppo $P_1$)\ \ } c_{c d}
\end{align}
\vspace{-5mm}
\caption{Partial order constraints for address $x$ in \myfig\ref{fig:iriw} on
SC} \label{fig:constraints-iriw} \end{figure}

\paragraph{Symbolic partial orders\label{partord}}
We associate each symbolic event $x$ of an \aes $\gamma$ with a unique
\emph{clock} variable $\clock{x}$ (\cf~\cite{lam78,sw10}) ranging over the
naturals.
For two events $x$ and $y$, we define the Boolean \emph{clock constraint} as
$\dfnshort{c_{xy}}{(\guard(x) \wedge \guard(y)) \implies \clock{x} <
\clock{y}}$ (``$<$'' being less-than over the integers).
We encode a relation $\textsf{r}$ over the symbolic events of $\gamma$ as the
formula $\phi(\textsf{r})$ defined as the conjunction of the clock
constraints $c_{xy}$ for all $(x,y) \in \textsf{r}$, \ie,
$\dfnshort{\phi(\textsf{r})}{\bigwedge_{(x,y) \in \textsf{r}} c_{xy}}$.

Let~$\Clocks$ be a valuation of the clocks of the events of~$\gamma$.
Let~$\Vals$ be a valuation of the symbols of the formula $\ssa$ associated to
$\gamma$. As noted in \mysec\ref{sec:concretisation}, $\Vals$ gives us concrete
values for the events of $\gamma$, and allows us to evaluate their guards. We
show below that $(\Clocks,\Vals)$ satisfies $\phi(\textsf{r})$ iff the
concretisation of $\textsf{r}$ \wrt $\Vals$ is acyclic, provided that this
relation has \emph{finite prefixes}.

A prefix of $x$ in a relation $\textsf{r}$ is a (possibly infinite) list
$S=[x_0, x_1, x_2, \dots]$ s.t.~$x=x_0$ and for all $i$, $(x_{i+1},x_i) \in
\textsf{r}$ (observe that the prefix is reversed \wrt the order imposed by the
relation).  The relation $\textsf{r}$ has finite prefixes if for each $x$,
there is a bound $l \in \mathbb{N}$ to the cardinality of the prefixes of
$x$ in $\textsf{r}$. We write $\card(S)$ for the cardinality of a list
$S=[x_0,x_1,x_2, \dots]$, \ie, $\dfnshort{\card(S)}{\card(\{x \mid \exists i.
x=x_i\})}$.  We write $\prefix(\textsf{r},x)$ for the set of prefixes of $x$ in
$\textsf{r}$.  Formally, $\textsf{r}$ has finite prefixes when $\forall x.
\exists l.  \forall S \in \prefix(\textsf{r},x). \card(S) < l$.  In our proofs
and in \myalg\ref{alg:ppo-new} we denote the concatenation of two lists~$S_1$
and~$S_2$ by $S_1\mo{++}S_2$.

In the following, we allow symbolic relations with infinite prefixes provided
their concretisations have finite prefixes. Thus we do not consider executions
with an infinite past, or running for more steps than the cardinality
of~$\mathbb{N}$. Our first lemma justifies why checking the acyclicity of a
concrete relation amounts to checking the satisfiability of the formula
encoding this relation symbolically: \begin{lemma} \label{sat:ac}
{$(\Clocks,\Vals)$ satisfies $\phi(\textsf{r})$ iff
$\conc(\textsf{r},\Vals)$ is acyclic and has finite prefixes.}
\end{lemma} \begin{proof} \noindent\em{$\implies$: We let
$\textsf{r}_c=\conc(\textsf{r},\Vals)$. One can show by induction that
$(*)$ if $(\Clocks,\Vals)$ satisfies $\phi(\textsf{r})$ then for all
$(x,y) \in \transc{\textsf{r}_c}$, $c_{xy}$ is true.
Now, suppose $\phi(\textsf{r})$ satisfied, and as a contradiction,
$\textsf{r}_c$ cyclic, \ie, $\exists x. (x,x) \in \transc{\textsf{r}_c}$. Thus
$c_{xx}$ is true by $(*)$; this contradicts the irreflexivity of $<$ over the
integers.

Now we show that $\textsf{r}_c$ has finite prefixes, \ie, for each $x$ we give
a bound $l$ over all $S \in \prefix(\textsf{r}_c,x)$. As a contradiction take
$S=[x_0,\dots x_n] \in \prefix(\textsf{r}_c,x)$ s.t.~$x=x_0$ and $\card(S) >
\clock{x}$. Thus for all $i$, we have $(x_{i+1}, x_i) \in \textsf{r}_c$ and
$\clock{x_{i+1}} < \clock{x_i}$ by $(*)$.  Since $n \geq \card(S)$, $\card(S) >
\clock{x}$ and $\clock{x_0}=\clock{x}$, we have $\clock{x_n} < 0$, which
contradicts the fact that our clocks are naturals.  Thus for each $x$ we can
take $l=\clock{x}$.

$\cimplies$:
Let $\textsf{r}_c=\conc(\textsf{r},\Vals)$. For all $e$ 
s.t.~$\guard(e)=\text{false}$, take $\clock{e}=0$. Thus $c_{xy}$ is true if
$\guard(x)$ or $\guard(y)$ is false. Now, have $(x,y) \in \textsf{r}$ with both
guards true, \ie, $(x,y) \in \textsf{r}_c$. Take $\clock{x}$ to be the maximal
cardinality of the $S$ in $\prefix(\textsf{r}_c,x)$, \emph{idem} for
$y$.  We want to prove $\clock{x} < \clock{y}$. Take $S$
s.t.~$\clock{x}=\card(S)$.  From $(x, y) \in \textsf{r}_c$, we have
$[y]\mo{++}S \in \prefix(\textsf{r}_c,y)$. Now, $\card([y]\mo{++}S)
\leq \clock{y}$ by maximality of $\clock{y}$. It suffices to prove $\card(S) <
\card([y]\mo{++}S)$.  Suppose $\card(S) \geq \card([y]\mo{++}S)$.
Then $y$ appears in $S$. Thus $(y,x) \in \transc{\textsf{r}_c}$ since $S$ is a
prefix of $x$; as $(x,y) \in \textsf{r}_c$ by hypothesis, we have a cycle in
$\textsf{r}_c$. } \end{proof}
The formula~$\phi(\textsf{r}_1 \cup \textsf{r}_2)$
is equivalent to~$\phi(\textsf{r}_1) \wedge \phi(\textsf{r}_2)$.
Thus we encode unions of relations, \eg, $\ghb_A\triangleq \ws \cup \fr \cup
\grf_A \cup \ppo_A \cup \ab_A$, as the conjunction of their respective
encodings.  By~\mylem\ref{sat:ac}, the acyclicity of~$\ghb_A$ corresponds to
the satisfiability of~$\phi(\ghb_s)$, where $\ghb_s$ is a symbolic encoding of
$\ghb_A$. To form~$\phi(\ghb_s)$, we form the conjunction of the
formulas~$\phi(\textsf{r})$, for~$\textsf{r}$ being a symbolic encoding
of $\ws$, $\fr$, $\grf_A$, $\ppo_A$ and $\ab_A$.
 
We now present these encodings, in that order. \mysec\ref{sec:model} also
relies on the program order per location for the $\uniproc$ check, and the
dependencies for the $\thin$ check, omitted for brevity.  We compute them
alongside the preserved program order; they use independent sets of clock
variables, but the same clock constraints.

We define auxiliaries over symbolic events: $\tr(e)$ is the thread identifier
of~$e$, $\loc(e)$ the memory address read from or written to (\eg, $x$
for~\llb{$(e)$}{R}{x}{y}), and $\val(e)$ its (symbolic) value.  Each algorithm
outputs constraints, whose conjunction we add to~$\parord$.

For each algorithm we state and prove a lemma about its correctness.
These follow the scheme of \mylem\ref{sat:ac}, \ie we show the encoding correct
for any satisfying valuation of clocks and \ssa.
Thus we will introduce symbolic encodings of sets $\textsf{r}(\gamma)$, where
membership in the set is given by a formula and thus depends on the actual
valuation under $\Clocks$ and $\Vals$.

\subsection{Read-from\label{pord:rf}}

\begin{algorithm}[!t]
\SetKwInOut{Input}{input}
\SetKw{Output}{output:}
\DontPrintSemicolon

\Input{$\gamma$, $A$ \quad \Output{$C_{\wf}, C_{\rf}$, $C_{\grf}$}}
\BlankLine
$\reads := \{ (\alpha, \{ r_1 \ldots r_n \} ) \mid \text{$r_i$ is read} \wedge \loc(r_i) = \alpha \}$\; \label{alg:rf:line1}
$\writes:=\{ (\alpha, \{ w_1 \ldots w_n \} ) \mid \text{$w_i$ is write}\wedge \loc(w_i) = \alpha \}$\; \label{alg:rf:line2}
$C_{\rf} := \emptyset$\label{alg:rf:line3a}; $C_{\grf} := \emptyset$; $C_{\wf} := \emptyset$\;\label{alg:rf:line3b}
\ForEach{$\alpha \text{ s.t. } \exists R,W. (\alpha, R) \in \reads \wedge (\alpha, W) \in \writes$}{\label{alg:rf:line4}
  \ForEach{$r \in R$}{\label{alg:rf:line5}
    $\rfsome := \emptyset$\label{alg:rf:line6}\; 
    \ForEach{$w \in W$}{\label{alg:rf:line7a}
      \If{$(r,w) \not\in \po(\gamma)$}{ \label{alg:rf:line7}
      $\rfsome := \rfsome \cup \{ s_{w r} \}$\;\label{alg:rf:line8}
      {$C_{\wf} := C_{\wf} \cup \{s_{wr} \implies (\guard(w) \wedge \sval(r) = \sval(w))\}$\;\label{alg:rf:line9} }
      $C_{\rf} := C_{\rf} \cup \{s_{w r} \Rightarrow c_{w r}\}$\;\label{alg:rf:line9b}
      \If{$(w,r)$ not relaxed on $A$ and $\tr(w) \neq \tr(r)$}{\label{alg:rf:line10a}
        $C_{\grf} := C_{\grf} \cup \{s_{w r} \Rightarrow c_{w r}\}$\;\label{alg:rf:line10} 
      } 
    }
  }
    {$C_{\wf} := C_{\wf} \cup \{\guard(r) \implies \bigvee_{s \in \rfsome} s\}$\;\label{alg:rf:line11}}
  }
}
\BlankLine
\caption{Constraints for read-from}
\label{alg:rf}
\end{algorithm}
For an architecture $A$ and an \aes~$\gamma$, \myalg\ref{alg:rf} encodes the
read-from (resp.~safe read-from) as the set of constraints $C_{\rf}$
(resp.~$C_{\grf}$). Following \mysec\ref{sec:model}, we add constraints
to $C_{\grf}$ depending on: first, the relation being within one thread or
between distinct threads (derivable from $\tr(w)$ and $\tr(r)$); second,
whether~$A$ exhibits store buffering, store atomicity relaxation, or both.

\myalg\ref{alg:rf} groups the reads and writes by address, in the
sets $\reads$ and $\writes$ (lines~\ref{alg:rf:line1} and~\ref{alg:rf:line2}).
For \ltest{iriw}, $\reads = \{ (\text{\lstinline{x}}, \{ a, d \}),
(\text{\lstinline{y}}, \{ b, c \}) \}$ and $\writes = \{ (\text{\lstinline{x}},
\{ i_0, e \}), (\text{\lstinline{y}}, \{ i_1, f\}) \}$.

The next step forms the potential read-from pairs. To that end,
\myalg\ref{alg:rf} introduces a free Boolean variable~$s_{wr}$ for each pair
$(w,r)$ of write and read to the same address (line~\ref{alg:rf:line8}), unless
such a pair contradicts program order (line~\ref{alg:rf:line7}). Indeed, if
$(w,r)$ is in $\rf$ and $(r,w)$ is in $\po$, this violates the $\uniproc$ check
of \mysec\ref{sec:model}. 

The variable $\rfsome$, initialised in line~\ref{alg:rf:line6}, collects the
variables~$s_{w r}$ in line~\ref{alg:rf:line8}. 
For \ltest{iriw}, the memory address \lstinline{x}, and the read~$a$, we have
$\rfsome = \{ s_{i_0 a}, s_{e a} \}$, \ie, the read $a$ can read either from
$i_0$ (the initial write to \lstinline{x}), or from the write $e$ on $P_2$.

Following \mysec\ref{sec:model}, each read must read from some write. We ensure
this at line~\ref{alg:rf:line11}, by gathering in $C_{\wf}$, for a given $r$,
the union of all the potential read-from $s_{wr}$ collected in $\rfsome$.

Going back to \ltest{iriw}, recall from \mysec\ref{sec:c-to-aes} that an event
has a guard indicating the branch of the program it comes from. In
\ltest{iriw}, the guard of~$a$ is true (as all the others), \ie, the read $a$
is concretely executed. Hence there exists a write (either $i_0$ or $e$) from
which $a$ reads, as expressed by the constraint $s_{i_0 a} \vee s_{e a}$ formed
at line~\ref{alg:rf:line11}. 

If $s_{wr}$ evaluates to true (\ie, $r$ reads from $w$), we record the
\emph{value constraint} $\val(r) = \val(w)$ in the set $C_{\wf}$
(line~\ref{alg:rf:line9}).  For \ltest{iriw}, we obtain the following for
\lstinline{x}: $(s_{i_0 a} \implies x_1 = x_0) \wedge\; (s_{i_0 d} \implies x_2
= x_0) \wedge\; (s_{e a} \implies x_1 = x_3) \wedge\; (s_{e d} \implies x_2 =
x_3)$.
The constraint $s_{i_0 a} \implies x_1 = x_0$ reads ``if $s_{i_{0}a}$ is true
(\ie, $a$ reads from $i_0$) then the value $x_1$ read by $a$ equals the value
$x_0$ written by $i_0$.''

The constraint added to $C_{\rf}$ is such that only if $s_{wr}$ evaluates to
true, the clock constraint $c_{wr}$ is enforced (line~\ref{alg:rf:line9b}).
For \ltest{iriw} we add the following to
$C_{\rf}$, for the address \lstinline{x}: $(s_{i_0 a} \implies c_{i_0 a})
\wedge\; (s_{e a} \implies c_{e a}) \wedge (s_{i_0 d} \implies c_{i_0 d})
\wedge\; (s_{e d} \implies c_{e d})$.

If $(w,r)$ is not relaxed on~$A$, we also add its clock constraint $c_{wr}$ to
$C_{\grf}$ (line~\ref{alg:rf:line10}). In \ltest{iriw}, all reads read from an
external thread. Thus on an architecture that does not relax store atomicity
(\ie, stronger than Power), we add the constraints that we added to $C_{\rf}$ to
$C_{\grf}$ as well.
On Power, $C_{\grf}$ remains empty.
 
We now write $\grf_A$ for both the function over concrete relations given by
the definition of $A$ as in \mysec\ref{sec:model}, and the corresponding
function over symbolic relations.
Given an architecture~$A$, we have $\grf_{A}(\textsf{r}) = \{(w,r) \in
\textsf{r} \mid (w,r) \text{ is not relaxed on $A$}\}$.
For example if $A$ is TSO, all thread-local read-from pairs are relaxed:
$\grf_{A}(\textsf{r}) = \{{(w,r) \in \textsf{r} \mid \tr(w) \neq \tr(r)}\}$.
We write $(w,r) \in \WR_{\alpha}$ when $w$ writes to an address $\alpha$ and $r$ reads
from the same $\alpha$, and $\dfnshort{\prf(\gamma)}{\{(w,r) \in \bigcup_{\alpha}\WR_{\alpha} \mid
(r,w) \not\in \po(\gamma)\}}$.  We write $\rf(\gamma)$ for the set $\{(w,r) \in
\prf(\gamma) \mid s_{wr}\}$ (with $s_{wr}$ of \myalg\ref{alg:rf}), and $\grf(\gamma)$ for
$\grf_A(\rf(\gamma))$. Note that we build the external safe read-from ($\grfe(\gamma)$)
only, \ie, between two events from distinct threads. We compute the
internal one as part of $\ppo_A$, in \myalg\ref{alg:ppo-new}.

Given an \aes $\gamma$, \myalg\ref{alg:rf} outputs $C_{\rf}, C_{\grf}$ and
$C_{\wf}$. 
Let $\WRv$ be a valuation of the $s_{wr}$ variables of $\gamma$.  We write
$\insta(\rf(\gamma),\WRv)$ (resp.~$\insta(\grf(\gamma),\WRv)$) for $\rf(\gamma)$
(resp.~$\grfe(\gamma)$) where $\WRv$ instantiates the $s_{wr}$ variables (thus
$\rf(\gamma)$ is a symbolic encoding of the set as noted before this
sub-section; we
use this notation similarly in the remainder of this section). We show that
\myalg\ref{alg:rf} gives the clock constraints encoding $\grf$ 
(we omit the corresponding lemma for $\rf$): 
\begin{lemma}\label{rf:sat}
{$(\Clocks,\Vals,\WRv)$ satisfies~$\bigwedge_{c \in C_{\wf}
\cup C_{\grf}} c$ iff~$(\Clocks,\Vals)$ satisfies \\ i) for all~$r$
s.t.~$\guard(r)$ is true, there is $w$ s.t.~$(w,r) \in \insta(\rf(\gamma),\WRv)$
and \\ ii) for all~$(w,r) \in \insta(\rf(\gamma),\WRv)$, $\guard(w)$ is true
and~$\val(w) = \val(r)$ and \\ iii) $\bigwedge_{(w,r) \in
\insta(\grfe(\gamma),\WRv)} c_{wr}$.} \end{lemma} \begin{proof}
\noindent\em{An induction on $R,W$ s.t.~$(\alpha,R) \in \reads$ and
$(\alpha,W) \in \writes$ for some address $\alpha$, then union for all $\alpha$
shows that $C_{\grf} = \{s_{wr} \implies c_{wr} \mid (w,r) \in \prf(\gamma) \cap
{\grfe_A}\}$, and $C_{\wf} = \bigcup_{\text{r is read}} \{\guard(r) \implies
\bigvee_{(w,r) \in \prf(\gamma)} s_{wr}\} \cup \{s_{wr} \implies (\guard(w) \wedge
\val(r) = \val(w)) \mid (w,r) \in \prf(\gamma) \}$. The first component of $C_{\wf}$ is
equivalent to i); the second to ii). $C_{\grf}$ is equivalent to iii).}
\end{proof}

The model described in \mysec\ref{sec:model} suggests that $\rf$ must be encoded
to be exclusive, \ie, to link a read to only one write.
An explicit encoding thereof, however, would be redundant, as
this is already enforced by $\ws$ and $\fr$.
Hence it suffices to consider \emph{at least one} write per read, as
\myalg\ref{alg:rf} does:
\begin{lemma}\label{rf:exclusive}
{$\uniproc(E,X) \implies \forall r. \neg(\exists w \neq w'. (w,r) \in \rf \wedge (w',r) \in \rf)$}
\end{lemma}
\begin{proof}
\noindent\em{By contradiction, have $w \neq w'$ s.t.~$(w,r) \in \rf$
and $(w', r) \in \rf$. By totality of $\ws$, $(w,w') \in \ws$ or $(w',w) \in
\ws$.  W.l.o.g.~have $(w,w') \in \ws$.  Then $(r,w') \in \fr$, \ie, a cycle in
$\rf \cup \fr$: $w',r,w'$, forbidden by $\uniproc$.} \end{proof}

\subsection{Write serialisation}
\label{sec:ws}
Given an \aes $\gamma$, \myalg\ref{alg:ws}
encodes the write serialisation $\ws$ as the set of constraints $C_{\ws}$.  By
definition, $\ws$ is a total order over writes to a given address.
\myalg\ref{alg:ws} implements the totality by ensuring that for two writes~$w
\neq w'$ to the same address either $c_{ww'}$ or $c_{w'w}$ holds. For implementation reasons we choose to express this as $\neg c_{ww'} \implies c_{w'w}$ rather than $c_{ww'} \vee c_{w'w}$.

\IncMargin{1em}
\begin{algorithm}[!t]
\SetKwInOut{Input}{input}
\SetKw{Output}{output:}
\DontPrintSemicolon

\Input{$\gamma$ \quad \Output{$C_{\ws}$}}
\BlankLine
$\writes := \{ (\alpha, \{ w_1\ldots w_n \} ) \mid \text{$w_i$ is write} \wedge \loc(w_i) = \alpha \}$\;
$C_{\ws} := \emptyset$; \ForEach{$\alpha \text{ s.t. } \exists W. (\alpha, W) \in \writes$}{\label{alg:ws:line2}
  \ForEach{$w \in W$}{\label{alg:ws:line3}
    \ForEach{$w' \in W \text{s.t.} \tr(w') \neq \tr(w)$}{\label{alg:ws:line4}
      $C_{\ws} := C_{\ws} \cup \{\neg c_{ww'} \Rightarrow c_{w'w}\}$\;\label{alg:ws:line5}
    }
  }
}
\BlankLine
\caption{Constraints for write serialisation}
\label{alg:ws}
\end{algorithm}
\DecMargin{1em}

\myalg\ref{alg:ws} groups the writes per address.  For each address $\alpha$ and
write~$w$ to $\alpha$ (lines~\ref{alg:ws:line2} and~\ref{alg:ws:line3}) we choose
another write~$w'$ to $\alpha$ (line~\ref{alg:ws:line4}), and build the disjunction
of clock constraints over $w$ and $w'$ (line~\ref{alg:ws:line5}). For
\ltest{iriw} we have $\writes = \{ (\text{\lstinline{x}}, \{ i_0, e \}),
(\text{\lstinline{y}}, \{ i_1, f\}) \}$, and the constraints: $(\neg c_{i_0 e}
\Rightarrow c_{e i_0}) \wedge (\neg c_{i_1 f} \Rightarrow c_{f i_1})$.

Note that we build the external $\ws$ only ($\ews$). With $\WW_\alpha$ the
pairs of writes to the address $\alpha$, and $\ws(\gamma)$ the set $\{(w,w') \in
\bigcup_\alpha WW_\alpha \mid c_{w'w} = \text{false}\}$, we have
$\ews(\gamma) \triangleq \ws(\gamma) \cap \{(w,w') \mid \tr(w) \neq
\tr(w')\}$. We compute the thread-local $\ws$ as part of $\ppo_A$, in
\myalg\ref{alg:ppo-new}. 
Given an input~$\gamma$ of \myalg\ref{alg:ws}, we now characterise the clock
constraints given by $C_{\ws}$.  Basically we show that \myalg\ref{alg:ws}
gives the clock constraints enconding $\ws$. The proof (omitted for brevity) is
by induction as for \mylem\ref{rf:sat}:
\begin{lemma}\label{ws:sat}{$(\Clocks,\Vals)$ satisfies
$\bigwedge_{c \in C_{\ws}} c$ iff it satisfies $\bigwedge_{(w,w') \in \ews}
c_{ww'}$.}  \end{lemma}

We quantify over all pairs of writes to the same address to build~$\ws$. Thus
for $w_0,w_1,w_2$ in $\ws$ in a concrete execution, we build $(w_0,w_1),
(w_1,w_2)$ and the redundant $(w_0,w_2)$ in the symbolic world. This is
inherent to the totality of~$\ws$.  

\subsection{From-read}
\label{sec:fr}
Given an \aes $\gamma$, \myalg\ref{alg:fr} encodes from-read as the set of
constraints~$C_{\fr}$. Recall that ${(r,w) \in \fr}$ means $\exists w'.  (w',r)
\in \rf \wedge (w',w) \in \ws$. 
The existential quantifier corresponds to a disjunction: $\bigvee_{w' \text{is
write}} (w',r) \in \rf \wedge (w',w) \in \ws$.  Since this disjunction can be
large, which is undesirable in the expression simplification used in
the implementation, we rewrite it as a
conjunction of small implications, each of which are simplified in isolation: $\bigwedge_{w' \text{is write}} \left( (r,w) \in
\fr \Leftarrow (w',r) \in \rf \wedge (w',w) \in \ws\right)$.
Thus \myalg\ref{alg:fr} encodes from-read as a conjunction of the premise
variables $s_{wr}$ of $C_{\rf}$ and clock variables $c_{ww'}$ of $C_{\ws}$
introduced in \myalg\ref{alg:rf} and~\ref{alg:ws}.

\IncMargin{1em}
\begin{algorithm}[!t]
\SetKwInOut{Input}{input}
\SetKw{Output}{output:}
\DontPrintSemicolon

\Input{$\gamma$ \quad \Output{$C_{\fr}$}}
\BlankLine
$\reads := \{ (\alpha, \{ r_1\ldots r_n \} ) \mid \text{$r_i$ is read} \wedge \loc(r_i) = \alpha \}$\;
$\writes := \{ (\alpha, \{ w_1\ldots w_n \} ) \mid \text{$w_i$ is write} \wedge \loc(w_i) = \alpha \}$\;
$C_{\fr} := \emptyset$\; 
\ForEach{$\alpha \text{ s.t. } \exists R,W. (\alpha, W) \in \writes, (\alpha, R) \in \reads$}{\label{alg:fr:line3}
  \ForEach{$(w, w') \in W\times W \text{s.t.} w' \neq w$}{\label{alg:fr:line4}
    \ForEach{$r \in R$ with $\tr(r) \neq \tr(w)$}{\label{alg:fr:line5}
    {$C_{\fr} := C_{\fr} \cup \{(s_{w' r} \wedge c_{w' w} \wedge \guard(w)) \implies c_{r w}\}$\;\label{alg:fr:line6}}
  }
}
}
\BlankLine
\caption{Constraints for from-read}
\label{alg:fr}
\end{algorithm}
\DecMargin{1em}

Again, we collect the sets of reads and writes per address. \myalg\ref{alg:fr}
considers triples $(w', w, r)$ of events to the same address, where $(w', w)$
is in the write serialisation, and $(w', r)$ is in read-from. We enumerate the
pairs of writes in line~\ref{alg:fr:line4}, and then pick a read in
line~\ref{alg:fr:line5}.  For each such triple we add in
line~\ref{alg:fr:line6} the clock constraint $c_{r w}$ under the premise that
i)~$(w', r) \in \rf$, witnessed by $s_{w' r}$, ii)~$(w', w) \in \ws$, witnessed
by $c_{w' w}$, and that iii)~the write~$w$ actually takes place in a concrete
execution, \ie, $\guard(w)$ evaluates to true.

For \ltest{iriw} all guards are true. For \lstinline{x}, we obtain: $(s_{i_0 a}
\wedge c_{i_0 e}) \implies c_{a e}) \wedge\; ((s_{i_0 d} \wedge c_{i_0 e})
\implies c_{d e}) \wedge\; ((s_{e a} \wedge c_{e i_0}) \implies c_{a i_0})
\wedge\; ((s_{e d} \wedge c_{e i_0}) \implies c_{d i_0}$.
For example, $(s_{i_0 a} \wedge c_{i_0 e}) \Rightarrow c_{a e}$, reads ``if
$s_{i_{0}a}$ is true (\ie, if $a$ reads from $i_0$), and if $c_{i_{0}e}$ is
true (\ie, $(i_0,e) \in \ws$) then $c_{ae}$ is true (\ie, $a$ is in $\fr$
before $e$).''

Given an \aes $\gamma$, \myalg\ref{alg:fr} outputs $C_{\fr}$. 
Note that we compute here the external from-read only ($\efr$), and the
internal one as part of $\ppo_A$, in \myalg\ref{alg:ppo-new}. We
show that \myalg\ref{alg:fr} gives the clock constraints encoding~$\fr$. The
(omitted) proof is as for \mylem\ref{rf:sat}:
\begin{lemma}\label{fr:sat}{
$(\Clocks,\Vals, \WRv)$ satisfies $\bigwedge_{c \in C_{\fr}} c$ iff 
$(\Clocks,\Vals)$ satisfies $\bigwedge_{(r,w) \in
\insta(\efr(\gamma),\WRv)} c_{rw}$.} \end{lemma}

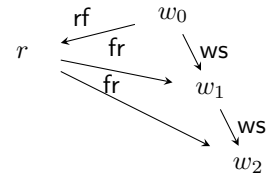
\begin{wrapfigure}[9]{r}{0.3\columnwidth}
  \vspace*{-2.5em}
  \centering
      \begin{tikzpicture}[>=stealth,thin,inner sep=0pt,text centered,shape=rectangle]

        \begin{scope}[minimum height=0.5cm,minimum width=0.5cm,text width=1.0cm]
          \node (w0)  at (0, 0)  {$w_0$};
          \node (r)  at (-2, -0.5)  {$r$};
          \node (w1)  at (0.5, -1)  {$w_1$};
          \node (w2)  at (1, -2)  {$w_2$};
        \end{scope}

        \path[->] (w0) edge node [above left=0.1cm] {$\rf$} (r);
        \path[->] (w0) edge node [right=0.1cm] {$\ws$} (w1);
        \path[->] (r) edge node [above=0.2cm] {$\fr$} (w1);
        \path[->] (r) edge node [above left=0.2cm] {$\fr$} (w2);
        \path[->] (w1) edge node [right=0.1cm] {$\ws$} (w2);
  \end{tikzpicture}
  \caption{\small{$\fr$~derives from~$\rf$ and~$\ws$}}
  \label{fig:fr}
\end{wrapfigure}
The $\fr$ defined above, together with $\ws$, does introduce possible
redundancies: given $(w_0,r) \in \rf$ with $(w_0, w_1) \in \ws$ and $(w_1, w_2)
\in \ws$, we have both $(r,w_1) \in \fr$ and $(r,w_2) \in \fr$ -- but the latter
is redundant as the same ordering is implied by $(r,w_1) \in \fr$ and $(w_1, w_2) \in
\ws$.
We could thus, instead, build a fragment of $\fr$, which we
write $\fr_0$. We define $\fr_0$ as $\{(r,w_1) \mid \exists w_0.
(w_0,r) \in \rf \wedge (w_0,w_1) \in \ws ~\wedge \nexists w'. ((w_0,w') \in
\ws \wedge (w',w_1) \in \ws)\}$. In \myfig\ref{fig:fr}, $(r,w_1)$ is in $\fr_0$
but not $(r,w_2)$, because there is a write ($w_1$) in $\ws$
between the write $w_0$ from which $r$ reads and $w_2$.  One can show that
$(r,w) \in \fr$ if $(r,w) \in \fr_0$ or there exists $w'$ s.t.~$(r,w') \in
\fr_0$ and $(w',w) \in \ws$, \ie, we can generate $\fr$ from $\fr_0$ and $\ws$.

\subsection{Preserved program order}
\label{sec:ppo}
For an architecture $A$ and an \aes~$\gamma$, \myalg\ref{alg:ppo-new} encodes
the preserved program order as the set $C_{\ppo}$.  In \mysec\ref{sec:model},
the function $\ppo_A$, which is part of the definition of $A$, determines if
$A$ relaxes a pair $(e,e')$ in program order in a concrete execution. For
example, RMO and Power relax read-read pairs, but PSO and stronger do not.

We reuse the notation $\ppo_A$ for the function collecting non-relaxed pairs in
symbolic program order. 
Unlike in~\mysec\ref{sec:model}, the non-relaxed pairs in symbolic program
order also include the internal safe read-from, internal write
serialisation, internal from-read, and the orderings due to Power's {\tt isync}
fence. We generate these constraints here, rather than
in~\myalg\ref{alg:rf}--\ref{alg:fr}, to limit the redundancies.  We
write~$\ppo_A(\gamma)$ for~$\ppo_A(\pobr(\gamma))$, or only~$\ppo_A$
if~$\gamma$ is clear from the context.  

\IncMargin{1em} \begin{algorithm}[!t] \SetKwInOut{Input}{input}
\SetKw{Output}{output:} \DontPrintSemicolon

\Input{$\gamma$, $A$ \quad \Output{$C_{\ppo}$}}
\BlankLine
$C_{\ppo} := \emptyset$; \ForEach{$S \in \po(\gamma) \wedge S \neq \emptyset$}{
  $S = [e]\mo{++}S'  \cap \{ e \mid \text{$e$ is not fence}\}$\;
  $\chains := [(e, \emptyset)]$\label{alg:ppo-new:line2}; $R := \text{true}$\; \label{alg:ppo-new:lineR}
\ForEach{$e' \in S'$}{\label{alg:ppo-new:line3}
    $T' := \emptyset$\label{alg:ppo-new:line4}\;
    \ForEach{$(e'', T'') \in \chains \text{ s.t. } \text{there is no } r \text{ s.t. }  \newline \phantom{xxxxxx} (e'', r) \in T' \text{ and }
    ((\guard(e') \wedge \guard(e'') \wedge R) \Rightarrow r)$}{\label{alg:ppo-new:line6new}
      $r_{e'' e'} := \notrelax A\; \gamma\; (e'',e')$\;\label{alg:ppo-new:line7new}
      \If{$r_{e'' e'} \text{is satisfiable}$}{\label{alg:ppo-new:line8new2}
        $C_{\ppo} := C_{\ppo} \cup \{r_{e'' e'} \Rightarrow c_{e'' e'}\}$\;\label{alg:ppo-new:line9new}
        $T' := T' \cup \{ (e'', r_{e'' e'}) \}$\;\label{alg:ppo-new:line10new2}
        \ForEach{$(e,r) \in T''$}{\label{alg:ppo-new:line11}
          \If{$\exists r'. (e, r') \in T'$}{
            $R := R \wedge (\rho \Leftrightarrow r' \vee (r_{e'' e'} \wedge r))$\label{alg:ppo-new:line12new2}\;
            $T' := \{ (e, \rho) \} \cup T' \setminus (e, r')$\;\label{alg:ppo-new:line14b}
          }
          \lElse{ $T' := \{ (e, r_{e'' e'} \wedge r) \} \cup T'$ }\label{alg:ppo-new:line15}
        }
      }
    }
    $\chains := [(e',T')]\mo{++}[\chains]$\;\label{alg:ppo-new:line14}
  }
}
\BlankLine
\caption{Constraints for preserved program order\label{alg:ppo-new}}
\end{algorithm}
\DecMargin{1em}

\myalg\ref{alg:ppo-new} avoids building redundant transitive closure
constraints, taking into account the guards of events: 
for two events $e_1,e_2$, we build a constraint iff $(e_1,e_2) \in
\ppo_A(\gamma)$. If, \eg, $\ppo_A(\pobr) = \pobr$ (on SC),
\myalg\ref{alg:ppo-new} creates constraints only for neighbouring events
in $\pobr(\gamma)$ in each control flow branch of the program.
 
As SSA and loop unrolling yield $\po(\gamma)$ (\ie, lists of symbolic events per thread)
rather than $\pobr(\gamma)$ (the corresponding DAG), we cannot construct
$C_{\ppo}$ by analysing control flow branches of the program.
Building $C_{\ppo}$ from $\po(\gamma)$ requires some more work.

To build $\ppo_{A}$, \myalg\ref{alg:ppo-new} uses the variable $\chains$, a
list of pairs~$(y, T)$. For a given $y$, its companion set~$T$ contains the
events $x$ occurring before $y$ in $\transc{\ppo_A}$ together with a formula
$r$ that characterises all paths of $\transc{\ppo_A}$ between $x$ and $y$. We
build $r$ from formulas $r_{e''e'}$ asserting that $(e'',e') \in \ppo_A$,
describing individual steps $(e'',e')$ of a path between $x$ and $y$.

We compute the formula $r_{e''e'}$ at line~\ref{alg:ppo-new:line7new}, using
the function $\notrelax$.  Given an \aes~$\gamma$ and a pair $(e'',e')$,
$\notrelax A \; \gamma\; (e'',e')$ returns a formula $r_{e'' e'}$ expressing
the condition under which $(e'',e')$ is not relaxed.  For PSO or stronger
models, $\notrelax$ only needs to take the direction of the events and their
addresses into account. For instance, TSO relaxes write-read pairs, but nothing
else. If a pair is necessarily relaxed, $\notrelax$ returns false, otherwise
$\notrelax A \; \gamma\; (e'',e') = \guard(e'') \wedge \guard(e')$.  For models
weaker than PSO, such as Alpha, RMO or Power, $\notrelax$ has to determine
data- and control dependencies, and handle Power's {\tt isync} fence. We
resolve data dependencies via a definition-use data flow analysis~\cite{asu86}
on the program part in program order between the two events. Control
dependencies use the data dependency analysis to test whether there exists a
branching instruction in program order between the events such that the
branching decision is in data dependency with the first event. For {\tt isync},
the approach is similar, except that in addition there must be an {\tt isync}
in program order between the branch and the second event.  We then add the
guard of the fence to the conjunction returned by $\notrelax$. 
 
For a given $e'$, we initialise its companion set $T'$ at
line~\ref{alg:ppo-new:line4}, then increment it in
lines~\ref{alg:ppo-new:line10new2}--\ref{alg:ppo-new:line15}. In
line~\ref{alg:ppo-new:line14b}, we use fresh variables~$\rho$ constrained in
the formula~$R$ (line~\ref{alg:ppo-new:line12new2}) to avoid
repeating sub-formulas, as is standard in, \eg, CNF
encodings~\cite{DBLP:conf/date/ChambersMV09}.  In
line~\ref{alg:ppo-new:line7new} we compute the condition~$r_{e'' e'}$ for
$(e'', e')$ not being relaxed on $A$ for each~$e''$ in $\chains$ (unless
skipped for transitivity, see below).  We generate the constraint $r_{e''e'}
\implies c_{e''e'}$ iff $r_{e'' e'}$ is satisfiable
(line~\ref{alg:ppo-new:line9new}), \ie, $(e'',e')$ is not relaxed on~$A$.


Now, suppose $e_1,e_2,e_3$ on the same thread all in $\ppo_A$; the companion
set of $e_2$ is $\{(e_1, r_{e_1 e_2})\}$, because $(e_1,e_2) \in \ppo_A$ and
there is no other event before $e_1$ on the thread.  Suppose that
\myalg\ref{alg:ppo-new} has already built the beginning of the chain formed by
$e_1$, $e_2$ and $e_3$, so that $\chains=[(e_2, \{ (e_1, r_{e_1 e_2}) \}),
(e_1, \emptyset)]$ (observe that the chains are in reverse order of $\po$).
At line~\ref{alg:ppo-new:line3}, for each remaining~$e'$ on a given thread,
\ie, $e_3$ in our example, \myalg\ref{alg:ppo-new} follows
lines~\ref{alg:ppo-new:line4}--\ref{alg:ppo-new:line9new} and adds a constraint
\wrt the immediate predecessor $e_2$ of~$e_3$ in $\ppo_A$. The subsequent
elements of $\chains$ ($e_1$ in our example) are also candidates for a clock
constraint. 

We do not add any constraint if $(e_1,e_3)$ is guaranteed to be in
$\transc{\ppo_A}$, as follows. Any remaining element of $\chains$ that belongs
to the companion set $T''$ of $e''$ is added to $T'$ at
lines~\ref{alg:ppo-new:line11}--\ref{alg:ppo-new:line15}. As an instance,
recall that $e_1$ is in the companion set of $e_2$. Thus, after generating the
constraint $c_{e_2 e_3}$ at line~\ref{alg:ppo-new:line9new}, we add $e_1$ with
its transitivity condition~$r_{e_2 e_3} \wedge r_{e_1 e_2}$ to $T'$ at
line~\ref{alg:ppo-new:line15}.  Then, line~\ref{alg:ppo-new:line6new} iterates
over the rest of $\chains$, \ie, $(e_1,\emptyset)$.  With the updated set~$T'$
the test $(e_1, r_{e_1}) \in T'$ yields $r_{e_1} =r_{e_2 e_3} \wedge r_{e_1
e_2}$, and thus amounts to checking the validity of $(\guard(e_3) \wedge
\guard(e_1)) \Rightarrow (r_{e_2 e_3} \wedge r_{e_1 e_2})$.  Remember that,
unless there is an {\tt isync}, the conditions $r_{xy}$ amount to conjunctions
over guards, hence in our example we are checking the validity of $(\guard(e_3)
\wedge \guard(e_1)) \Rightarrow \guard(e_3) \wedge \guard(e_2) \wedge
\guard(e_1)$.  If all three events $e_1$, $e_2$ and $e_3$ are on the same
control flow branch, the implication is valid because all guards are equal.
This makes the test of line~\ref{alg:ppo-new:line6new} fail and $(e_1, e_3)$
will not be considered for adding another constraint $c_{e_1 e_3}$, which would
have been redundant.  When the implication is not valid, the test of
line~\ref{alg:ppo-new:line6new} succeeds; then we add another constraint
$c_{e_1 e_3}$, as this is not redundant here.

This elaboration on guards is essential as witnessed by the following variant of
our example: assume, in contrast to the above, that $e_2$ is not a dominator of
$e_3$ on the control flow graph. This might occur in a program fragment {\tt
(if $e_1$ then $e_2$); $e_3$}, where the guard of~$e_2$ would be
different from that of~$e_1$ or~$e_3$.  If we were to skip $(e_1, e_3)$ as
above, the constraints would be insufficient to enforce the order of $e_1$
before $e_3$ when $\guard(e_2)$ evaluates to false. In this case, the premises
$\guard(e_1) \wedge \guard(e_2)$ and $\guard(e_2) \wedge \guard(e_3)$ of
$c_{e_1 e_2}$ and $c_{e_2 e_3}$, respectively, are false, hence the clock
constraints $\clock{e_1} < \clock{e_2}$ and $\clock{e_2} < \clock{e_3}$ are not
enforced, leaving the order of $(e_1, e_3)$ unconstrained.

We illustrate \myalg\ref{alg:ppo-new} on the \aes $\gamma$ of \ltest{iriw} (\cf
\myfig\ref{ssa:iriw}). \myalg\ref{alg:ppo-new} proceeds over $\po(\gamma)$,
equal to $\{ [i_0, i_1], [a, b], [c, d], [e], [f]\}$ for \ltest{iriw}.  Given a
non-empty list $S$ of $\po(\gamma)$, \eg, $S=[a,b]$ corresponding to $P_0$, the
first non-fence event~$a$ of~$S$ initialises at line~\ref{alg:ppo-new:line2}
the variable $\chains$ (explained below in detail).  The loop at
line~\ref{alg:ppo-new:line3} proceeds with the tail $S'$ of the list~$S$.  Thus
for $P_0$ at this point we have $\chains = [(a, \emptyset)]$ and
\myalg\ref{alg:ppo-new} proceeds with $S' = [b]$. 

The contents of $\chains$ depend on the architecture $A$, as \ltest{iriw}
shows.  For $P_0$, recall that $\chains = [(a,\emptyset)]$ and
only $b$ remains in $S'$.  If $A$ relaxes read-read pairs, \eg, RMO or weaker,
then $(a,b)$ is relaxed. Thus we do not add any clock constraint to
$C_{\ppo}$ at line~\ref{alg:ppo-new:line9new} and eventually $\chains =
[(b,\emptyset), (a,\emptyset)]$ in line~\ref{alg:ppo-new:line14}.
If $A$ does not relax read-read pairs, \eg, PSO or stronger, we add $c_{a b}$
to $C_{\ppo}$ at line~\ref{alg:ppo-new:line9new} and add~$a$ with the guard
conjunction $\text{true}$ to~$T'$ at
line~\ref{alg:ppo-new:line15}.  Thus $\chains = [(b, \{(a, \text{true})\}),
(a,\emptyset)]$.
Let us now
characterise the output of \myalg\ref{alg:ppo-new}, given an input $\aes$
$\gamma$:
\begin{lemma}\label{ppo:charac}
{\myalg\ref{alg:ppo-new} outputs
$\{r_{xy} \implies c_{xy} \mid (x,y) \in \ppo_A\}$.}
\end{lemma} 
\begin{proof}
\noindent\em{We write $L_1$ (resp.~$L_2$) for the loop from
line~\ref{alg:ppo-new:line3} to~\ref{alg:ppo-new:line14}
(resp.~\ref{alg:ppo-new:line6new} to~\ref{alg:ppo-new:line15}).  $L_1$
maintains the invariant that 
$S=\revdel(\chains)\mo{++}S'$, where $\revdel$ reverses its argument and
deletes $T$ for each element $(e,T)$ of its argument. We write
$\wpath_{x,y}^{e}(e_1,\dots,e_n)$ when there is a path from $x$ to $y$ in
$\ppo_A(\gamma)$ passing by $e$, \ie, $e_1=x$ and $e_n=y$ and $\forall i.
(e_i,e_{i+1}) \in \ppo_A(\gamma)$ and $\exists i. e_i=e$. $L_2$ maintains the
invariant that $T'=\bigcup_{e \in [e'',e']} T_e^{e'}$, where $e \in [e'',e']$
means $(e'',e) \in \po \wedge (e,e') \in \po$, and $T_e^{e'} = \{(x,r_x) \mid
r_x = \bigvee_{\wpath_{x,y}^{e}(e_1,\dots,e_n)} \bigwedge_{1 \leq i \leq n}
r_{e_i e_{i+1}}\}$.  We conclude by double inclusion of $C_{\ppo}$ and
$\{r_{xy} \implies c_{xy} \mid (x,y) \in \ppo_A\}$, omitted for
brevity.}\end{proof}
Since the $r_{xy}$ are guard conditions, we just need to evaluate the guards to
evaluate them. We show that \myalg\ref{alg:ppo-new} gives the clock constraints
encoding $\ppo$; the proof is immediate by \mylem\ref{ppo:charac}:
\begin{lemma}\label{ppo:sat}{$(\Clocks,\Vals)$ satisfies
$\bigwedge_{c \in C_{\ppo}} c$ iff it satisfies $\bigwedge_{(x,y) \in \ppo_A}
c_{xy}$.}
\end{lemma} 

\subsection{Memory fences and cumulativity}
\label{sec:ab}
\begin{algorithm}[!t]
\SetKwInOut{Input}{input}
\SetKw{Output}{output:}
\DontPrintSemicolon

\Input{$\gamma$, $A$ \quad \Output{$C_{\ab'}$}}
\BlankLine
$C_{\ab'} := \emptyset$; \ForEach{$S \in \po(\gamma) \wedge S \neq \emptyset$}{\label{alg:ab:line1}
  $\fences := \{ s \mid s \in S \wedge \text{s is fence} \}$\;\label{alg:ab:line2}
  \ForEach{$e \in S \setminus \fences$}{\label{alg:ab:line3}
    \ForEach{$s \in \fences$}{\label{alg:ab:line4}
      \uIf{$(e,s) \in \po(\gamma)$}{\label{alg:ab:line5}
        $C_{\ab'} := C_{\ab'} \cup \{\guard(s) \Rightarrow  c_{e s}\}$\;\label{alg:ab:line6}
        \If{$A$ is not store atomic}
        {\ForEach{$(w,e) \text{ being a w-r pair s.t.}$ 
                 \phantom{xx}$\loc(w) = \loc(e)$ and \phantom{xx}$\tr(w) \neq \tr(e)$}{\label{alg:ab:line8}
          {$C_{\ab'} := C_{\ab'} \cup \{(\guard(s) \wedge s_{w e}) \implies c_{w s}\}$\;\label{alg:ab:line9}}}
        }
      }\lElse{$C_{\ab'} : =C_{\ab'} \cup \{\guard(s) \implies c_{se}\}$\;\label{alg:ab:line10}
        \Indentp{4ex}\If{$A$ is not store atomic}
        {\ForEach{$(e,r) \text{ being a w-r pair s.t.}$
                  \phantom{xx}$\loc(e) = \loc(r)$ and \phantom{xx}$\tr(e) \neq \tr(r)$}{\label{alg:ab:line13}
          {$C_{\ab'} := C_{\ab'} \cup \{(\guard(s) \wedge s_{er}) \implies c_{sr}\}$\;\label{alg:ab:line14}}}
        }
      }
    }
  }
}
\BlankLine
\caption{Constraints for memory fences}
\label{alg:ab}
\end{algorithm}
Given an architecture $A$ and an \aes $\gamma$, \myalg\ref{alg:ab} encodes the
fence orderings as the set $C_{\ab'}$. A fence~$s$ potentially induces
orderings over all $(e, e')$ s.t.~$e$ is in $\po$ before~$s$ and~$e'$
after~$s$, which is quadratic in the number of events in $\po$ for each fence.
Cumulativity constraints depend on the read-from to appear in the
concrete event structure, and again these are paired with all events before or after (in $\po$) a fence. 
We alleviate this with the fence events (see below).  The implementation
supports x86's \texttt{mfence} and Power's \texttt{sync}, \texttt{lwsync} and
\texttt{isync}. We handle \texttt{isync} as part of $\ppo$ in
\myalg\ref{alg:ppo-new}. We first present x86's \texttt{mfence} and Power's
\texttt{sync}, then \texttt{lwsync}. 

\paragraph{Fences \texttt{mfence} and \texttt{sync}}
\myalg\ref{alg:ab} applies its procedure to $\po(\gamma)$
(line~\ref{alg:ab:line1}). For example, assume \texttt{sync} fences between the
read-read pairs of $P_0$ and $P_1$ of \ltest{iriw}, associated with the fences
events $s_0$ and $s_1$.  We then have $\po(\gamma) = \{ [i_0, i_1], [a, s_0,
b], [c, s_1, d], [e], [f] \}$.

For each list~$S$ of $\po(\gamma)$ (\ie, per thread), we compute at
line~\ref{alg:ab:line2} the set $\fences$, containing the fence events of~$S$.
For \ltest{iriw}, $\fences$ is empty for $P_2$ and $P_3$.  For $P_0$, we have
$\fences = \{ s_0 \}$, and $\{ s_1 \}$ for $P_1$.  We test at
line~\ref{alg:ab:line5} for each pair $(e,s)$ s.t.  $e$ is a non-fence event
and $s$ is fence whether $(e,s)$ is in program order, or rather $(s,e)$. We then
build the according non-cumulative constraints, and constraints for
A-cumulativity (for $(e,s)$ in program order) or B-cumulativity (otherwise).

For non-cumulativity, if $e$ is before (resp.~after) $s$ in program order,
\myalg\ref{alg:ab} produces at line~\ref{alg:ab:line6} the clock constraint
$c_{es}$ (resp.~$c_{se}$ at line~\ref{alg:ab:line10}). In \ltest{iriw}, all
guards are true, hence we generate $c_{a s_0}$ (resp.~$c_{c s_1}$) for the
event $a$ (resp.~$c$) in $\po$ before the fence $s_0$ (resp.~$s_1$) on $P_0$
(resp.~$P_1$).  Line~\ref{alg:ab:line10} generates $c_{s_0 b}$
(resp.~$c_{s_1 d}$) for~$b$ (resp.~$d$), in $\po$ after the fence $s_0$
(resp.~$s_1$) on $P_0$ (resp.~$P_1$).

If $A$ relaxes store atomicity, we build cumulativity constraints.  For
A-cumulativity, \myalg\ref{alg:ab} adds at line~\ref{alg:ab:line9} the
constraint $s_{we} \implies c_{ws}$, for each~$(w,e)$ s.t.~$e$ is in $\po$
before the fence $s$, and $e$ reads from the write $w$. The constraint reads
``if $\guard(s)$ is true (\ie, the fence is concretely executed) and if
$s_{we}$ is true (\ie, $e$ reads from $w$), then $c_{ws}$ is true (\ie, there
is a global ordering, due to the fence $s$, from $w$ to $s$)''. All other
constraints, \ie, the actual ordering of~$w$ before some event~$e'$ in $\po$
\emph{after}~$s$, follow by transitivity.
We handle B-cumulativity in a similar way, given in
lines \ref{alg:ab:line13} and \ref{alg:ab:line14}.

As Power relaxes store atomicity, the {\tt sync} fences between the read-read
pairs of \ltest{iriw} create A-cumulativity constraints, namely for $s_0$ (and
analogous ones for $s_1$): $(s_{i_0 a} \Rightarrow c_{i_0 s_0}) \wedge (s_{e a}
\Rightarrow c_{e s_0})$.  
 
If we were not using fence events, we would create a clock constraint $c_{we'}$
for every~$e'$ in program order after the fence~$s$ to implement
\mysec\ref{sec:model}, for each fence~$s$. Thus the non-cumulative part would be
cubic already, whereas fence events yield a quadratic number at most.
For cumulativity, we would obtain a constraint for every pair $(r, e')$
s.t.~$(w,r) \in \rfe$ and~$r$ is in $\po$ before the fence~$s$. The resulting
number of constraints is the number of such pairs $(r,e')$ times the number of
pairs $(w,r)$, \ie, cubic in the number of events \emph{per fence}~$s$.
Furthermore cases of both A- and B-cumulativity at the same fence~$s$ need to be
taken into account, resulting in even higher complexity. Fence events,
however, reduce all these cases, including the combined one, to cubic complexity
(all triples of external writes, reads, and fence events).


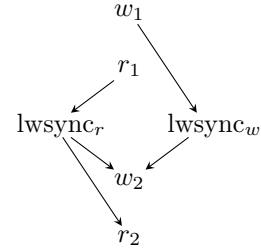
\begin{wrapfigure}[11]{r}{0.3\columnwidth}
  \vspace*{-3em}
  \centering
      \begin{tikzpicture}[>=stealth,thin,inner sep=0pt,text centered,shape=rectangle]

        \begin{scope}[minimum height=0.3cm,minimum width=0.5cm,text width=1.0cm]
          \node (w1)  at (0, 0)  {$w_1$};
          \node (r1)  at (0, -0.75)  {$r_1$};
          \node (lwsyncr)  at (-1, -1.5)  {lwsync$_r$};
          \node (lwsyncw)  at (1, -1.5)  {lwsync$_w$};
          \node (w2)  at (0, -2.25)  {$w_2$};
          \node (r2)  at (0, -3)  {$r_2$};
        \end{scope}

        \path[->] (w1) edge (lwsyncw);
        \path[->] (r1) edge (lwsyncr);
        \path[->] (lwsyncr) edge (r2);
        \path[->] (lwsyncr) edge (w2);
        \path[->] (lwsyncw) edge (w2);
  \end{tikzpicture}
\caption{Constraints for {\tt lwsync}}
  \label{fig:lwsync}
\end{wrapfigure}
\paragraph{Fence \texttt{lwsync}}
As {\tt lwsync} does not order write-read pairs (\cf \mysec\ref{sec:model}), we
need to avoid creating a constraint~$c_{wr}$ between a write~$w$ and a read~$r$
separated by an {\tt lwsync}.  To do so, we use two distinct clock
variables~$\clock{s}^r$ and~$\clock{s}^{w}$ for an~\texttt{lwsync}~$s$.  This
avoids the wrong transitive constraint~$c_{wr}$ implied by~$c_{ws}$
and~$c_{sr}$.
\myfig\ref{fig:lwsync} illustrates this setup: the write-read pair $(w_1, r_2)$
will not be ordered by any of the constraints, but all other pairs are ordered.

To create a clock constraint in
lines~\ref{alg:ab:line6},~\ref{alg:ab:line9},~\ref{alg:ab:line10},
or~\ref{alg:ab:line14}, we then pick one or both of the clock variables, as
follows.
If~$e$ is a read, the clock constraint is $\clock{e} < \clock{s}^r$ when $e$ is
before $s$, \ie, lines~\ref{alg:ab:line6} or~\ref{alg:ab:line9} (or $\clock{s}^r
< \clock{e}$ if $e$ is after, \ie, lines~\ref{alg:ab:line10}
or~\ref{alg:ab:line14}).  If~$e$ is a write preceding~$s$ (\ie,
lines~\ref{alg:ab:line6} or~\ref{alg:ab:line9}), the clock constraint is
$\clock{e} < \clock{s}^{w}$. Finally, if~$e$ is a write after~$s$, \ie,
lines~\ref{alg:ab:line10} or~\ref{alg:ab:line14}, the clock constraint is the
conjunction $(\clock{s}^{w} < \clock{e}) \wedge (\clock{s}^{r} < \clock{e})$.
To make {\tt lwsync} non-cumulative (\cf footnote in \mysec\ref{sec:model}), we
just need to disable the lines
\ref{alg:ab:line8},\ref{alg:ab:line9},\ref{alg:ab:line13} and
\ref{alg:ab:line14}.

In \ltest{iriw}, if we use \texttt{lwsync} instead of \texttt{sync} as
discussed above, we obtain the following constraints:
$(\clock{a} < \clock{s_0}^r) \wedge (\clock{s_0}^r < \clock{b}) \wedge
(s_{i_0 a} \implies \clock{i_0} < \clock{s_0}^{w}) \wedge (s_{e a} \implies
\clock{e} < \clock{s_0}^{w})$.
These constraints will \emph{not} order the writes~$i_0$ or~$e$ with the
read~$b$, because~$i_0$ and~$e$ are ordered \wrt to $\clock{s_0}^{w}$, but~$b$
is only ordered \wrt the distinct $\clock{s_0}^r$.  This corresponds to the
fact that placing {\tt lwsync} fences in \ltest{iriw} does not forbid the
non-SC execution. 

%
Given an \aes $\gamma$, \myalg\ref{alg:ab} outputs $C_{\ab'}$. 
We let $\rfe(\gamma)$ be $\{(w,r) \in \bigcup_\alpha \WR_\alpha \mid \tr(w)\neq
\tr(r) \wedge s_{wr}\}$. We write $\ab'(\gamma)$ for $\{(e_1,e_2) \mid
\nc'(e_1,e_2) \vee \ac'(e_1,e_2) \vee \bc'(e_1,e_2)\}$, where $\nc'(e_1,e_2)$
corresponds to non-cumulativity, \ie, $(e_1,e_2) \in \po(\gamma) \wedge
((\guard(e_1) \wedge e_1 \text{ is fence}) \vee (\guard(e_2) \wedge e_2 \text{
is fence})) \wedge \text{not both $e_1$ and $e_2$ are fences}$, $\ac'(e_1,e_2)$
to A-cumulativity, \ie, $\exists r. (e_1,r) \in \rfe(\gamma) \wedge (r,e_2) \in
\po(\gamma) \wedge \guard(e_2) \wedge e_2 \text{ is fence}$, and
$\bc'(e_1,e_2)$ corresponds to B-cumulativity, \ie, $\exists w.  (w,e_2) \in
\rfe(\gamma) \wedge (e_1,w) \in \po(\gamma) \wedge \guard(e_1) \wedge e_1
\text{ is fence}$. 
We show that \myalg\ref{alg:ab} gives the clock constraints encoding
$\ab'$. The proof is immediate like for \mylem\ref{rf:sat}:
\begin{lemma}\label{ab':sat}{$(\Clocks,\Vals,\WRv)$
satisfies $\bigwedge_{c \in C_{\ab'}} c$ iff $(\Clocks,\Vals)$
satisfies $\bigwedge_{(x,y) \in \insta(\ab'(\gamma),\WRv)} c_{xy}$.}
\end{lemma}

%
We let~$\ab(\gamma)$ be the symbolic version of~$\ab$ in~\mysec\ref{sec:model},
\ie, we let $\nc(e_1,s,e_2)$ be~$\guard(s) \wedge s \text{ is fence } \wedge
(e_1,s) \in \po(\gamma) \wedge (s,e_2) \in \po(\gamma)$, $\ac(e_1,s,e_2)$
be~$\exists r.  (e_1,r) \in \rfe(\gamma) \wedge \nc(r,s,e_2)$
and~$\bc(e_1,s,e_2)$ be $\exists w. \nc(e_1,s,w) \wedge (w,e_2) \in
\rfe(\gamma)$. We only prove this encoding sound \wrt~\mysec\ref{sec:model},
as~$\ab'$ is more fine-grained than $\ab$ (to see why, note that one cannot
express $\nc'(e_1,e_2)$ as a combination of $\nc, \ac$ or $\bc$). Yet we prove
our overall encoding complete in~\mythm\ref{sat:ac:ghb}.  
\begin{lemma}\label{ab:sat} {If $(\Clocks,\Vals,\WRv)$
satisfies $\bigwedge_{c \in C_{\ab'}} c$ then $(\Clocks,\Vals)$
satisfies $\bigwedge_{(e_1,e_2) \in \insta(\ab(\gamma),\WRv)} c_{e_1 e_2}$.}
\end{lemma}
\begin{proof}\noindent\em{ We give only the case $\lwsync(\gamma)$.
Take $(e_1,e_2) \in \lwsync(\gamma)$, \ie, there is an {\tt lwsync} $s$
s.t.~$\nc(e_1,s,e_2)$ or $\ac(e_1,s,e_2)$ or $\bc(e_1,s,e_2)$. In the $\nc$
case, we know that $s$ is a fence and $\guard(s)$ is true, and $(e_1,s) \in
\po(\gamma)$ and $(s,e_2) \in \po(\gamma)$, \ie, $\nc'(e_1,s)$ and
$\nc'(s,e_2)$. Thus $c_{e_1 s}$ and $c_{se_2}$ are in $C_{\ab'}$. Now, by
definition of $\lwsync(\gamma)$, $(e_1,e_2) \not\in \WR$. For $(e_1,e_2) \in
\WW$, $c_{e_1 s}$ is $\clock{e_1} < \clock{s}^{w}$ and $c_{se_2}$ is
$\clock{s}^{w} < \clock{e_2}$.  Thus $\clock{e_1} < \clock{e_2}$, \ie, $c_{e_1
e_2}$ holds.  Writing $\RR$ for the read-read pairs, take $(e_1,e_2) \in \RR$.
Thus $c_{e_1 s}$ is $\clock{e_1} < \clock{s}^{r}$ and $c_{s e_2}$ is
$\clock{s}^{r} < \clock{e_2}$. Hence $\clock{e_1} < \clock{e_2}$, \ie, $c_{e_1
e_2}$ holds. For $(e_1,e_2) \in \RW$, $c_{e_1 s}$ is $\clock{e_1} <
\clock{s}^{r}$ and $c_{s e_2}$ is $(\clock{s}^{w} < \clock{e_2}) \wedge
(\clock{s}^{r} < \clock{e_2})$.  Thus $\clock{e_1} < \clock{e_2}$, \ie, $c_{e_1
e_2}$ holds.  In the~$\ac$ case, $s$ is a fence and~$\guard(s)$ is true, and
there is~$r$ s.t.~$(e_1,r) \in \rfe(\gamma)$ and~$\nc(r,s,e_2)$. Thus~$e_1$ is
a write (source of a~$\rf$), and since~$(e_1,e_2) \not\in \WR$ (by definition
of $\lwsync(\gamma)$), $e_2$ is a write. So~$\ac'(e_1,s)$, and~$\nc'(s,e_2)$,
\ie, $c_{e_1 s}$ and~$c_{s e_2}$ hold. We are back to the~$\WW$ case.  In
the~$\bc$ case, $s$ is a fence and~$\guard(s)$ is true, and there is~$w$
s.t.~$\nc(e_1,s,w)$ and~$(w,e_2) \in \rfe(\gamma)$.  Thus~$e_2$ is a read
(target of a~$\rf$), and since~$(e_1,e_2) \not\in \WR$, $e_1$ is a read.
So~$\nc'(e_1,s)$ and~$\bc'(s,e_2)$, \ie, $c_{e_1 s}$ and~$c_{s e_2}$ hold. We
are back to the~$\RR$ case.} \end{proof}
 
\subsection{Soundness and completeness of the encoding}

Given an architecture $A$ and a program, the procedure of
\mysec\ref{sec:c-to-aes} and \mysec\ref{sec:aes-to-ces} outputs a formula~$\ssa
\wedge \parord$ and an \aes~$\gamma$. This formula provably encodes the
executions of this program valid on~$A$ and violating the property encoded in
$\ssa$ in a sound and complete way.  Proving this requires proving that any
assignment to the system corresponds to a valid execution of the program, and
vice versa. This result requires three steps, one for $\uniproc$, one for
$\thin$ and one for the acyclicity of $\ghb$. By lack of space, we show only
the last one.  Given an~\aes~$\gamma$, we write $\phi$ for $\bigwedge_{c \in
C_{\ppo} \cup C_{\grf} \cup C_{\wf} \cup C_{\ws} \cup C_{\ab'}} c$:
\begin{theorem}\label{sat:ac:ghb} {The formula~$\ssa \wedge \phi$ is satisfiable
iff there are $\Vals$, a valuation to the symbols of \ssa, and a well formed $X$
s.t.~$\ghb_A(\conc(\gamma,\Vals),X)$ is acyclic and has finite prefixes.}
\end{theorem} \begin{proof} \noindent\em{Let
$(\Clocks,\Vals,\WRv)$ be a satisfying assignment of $\ssa \wedge \phi$. By
\mylem\ref{ppo:sat}, \ref{rf:sat}, \ref{ws:sat}, \ref{fr:sat} and \ref{ab':sat},
we know that $(\Clocks,\Vals,\WRv)$ satisfies $\phi$ iff i) for all $r$
s.t.~$\guard(r)$ is true, there is $w$ s.t.~$(w,r) \in \insta(\rf(\gamma),\WRv)$
and ii) for all $(w,r) \in \insta(\rf(\gamma),\WRv)$, $\guard(w)$ is true and
$\val(w)=\val(r)$ and iii) $(\Clocks,\Vals,\WRv)$
satisfies~$\phi(\ppo_A(\gamma)) \wedge \phi(\insta(\grf(\gamma),\WRv)) \wedge
\phi(\ews(\gamma)) \wedge \phi(\insta(\fr(\gamma),\WRv)) \wedge
\phi(\insta(\ab'(\gamma),\WRv))$.
 
$\implies$: 
Take $X=(\conc(\insta(\rf(\gamma),\WRv),\Vals)),(\conc(\ws(\gamma),\Vals))$.
Note that i) and ii), together with \mylem\ref{rf:exclusive}, imply that
$\rf(X)$ is well formed. For $\ws(X)$, this comes from the totality of
$\ws(\gamma)$ over writes to the same address, implied by the shape of
$C_{\ws}$ (\cf \myalg\ref{alg:ws}) for the external $\ws$, and by the totality
of $\po(\gamma)$ for the internal $\ws$.  

By \mylem\ref{ab':sat} and \ref{ab:sat}, iii) says that
$(\Clocks,\Vals,\WRv)$ satisfies $\phi(\textsf{r})$, with
$\textsf{r}=\ppo_A(\gamma) \cup \insta(\grf_A(\gamma),\WRv) \cup \ws(\gamma)
\cup \insta(\fr(\gamma),\WRv) \cup \insta(\ab(\gamma),\WRv)$. By \mylem\ref{sat:ac}, since
$\ghb_A(\conc(\gamma,\Vals),X)$ is $\conc(\textsf{r},\Vals)$, we have our
result.
%

$\cimplies$: 
We let $E$ be $\conc(\gamma,\Vals)$.  Take $\WRv$ s.t.~$s_{wr}$ is true iff
$(w,r) \in \rf(X)$.  $\rf(X)$ being well-formed implies i) and ii). 

We let $\ghb'_{A}(E,X)$ be $\textsf{r}_1 \cup \ab'(E,X)$, with
$\textsf{r}_1=\ppo_A(E) \cup \grf_A(X) \cup \ws(X) \cup \fr(E,X)$. Note that
$\ghb_{A}(E,X)$ is $\textsf{r}_1 \cup \ab(E,X)$. We show below that the
acyclicity of $\ghb_{A}(E,X)$ implies the acyclicity of $\ghb'_{A}(E,X)$
(\emph{idem} for finite prefixes).  Then we take $\textsf{r}=\ghb'_{A}(E,X)$ in
\mylem\ref{sat:ac}, and take $\Clocks$ as in \mylem\ref{sat:ac}. Hence we have
$(\Clocks,\Vals)$ satisfying $\phi(\ppo_A(\gamma)) \wedge
\phi(\insta(\grf(\gamma),\WRv)) \wedge \phi(\ews(\gamma)) \wedge
\phi(\insta(\fr(\gamma),\WRv)) \wedge \phi(\insta(\ab'(\gamma),\WRv))$, namely
iii); our result follows.

We let $\textsf{r}_2=\{(e_1,e_2) \in \ab';\ab' \mid \text{neither } e_1 \text{
nor } e_2 \text{ is a fence}\})$. We write $(x,y) \in \textsf{r};\textsf{r'}$
for $\exists z. (x,z) \in \textsf{r} \wedge (z,y) \in \textsf{r'}$. 
%
One can show: $(*)$ if $(e_1,e_2) \in \transc{\ghb'}$, we have $(e_1,e_2)$
in $\ab';\transc{(\transc{\textsf{r}_1} \cup \transc{\textsf{r}_2})}$, or
$\transc{(\transc{\textsf{r}_1} \cup \transc{\textsf{r}_2})};\ab'$, or
$\transc{(\transc{\textsf{r}_1} \cup \transc{\textsf{r}_2})}$. 

Acyclicity: by contradiction, take a cycle in $\ghb'_{A}(E,X)$, \ie, $x$
s.t.~$(x,x) \in \transc{(\ghb'_{A}(E,X))}$. In the first two cases $(*)$,
$\ab'$ connects two non-fence events, a contradiction.  Hence a cycle in
$\ghb'$ implies one in $\textsf{r}_1 \cup \textsf{r}_2$, \ie, in $\ghb$
since $\textsf{r}_2 \subseteq \ab$.

Prefixes: as a contradiction, take an infinite path in
$\transc{\ghb'_{A}(E,X)}$.  Only the cases $\transc{(\transc{\textsf{r}_1} \cup
\transc{\textsf{r}_2})}$ and $\ab';\transc{(\transc{\textsf{r}_1} \cup
\transc{\textsf{r}_2})}$ of $(*)$ apply, and both imply an infinite path in
$\transc{(\transc{\textsf{r}_1} \cup \transc{\textsf{r}_2})}$.  Hence, we have
an infinite prefix in $\transc{\ghb}$, since $\textsf{r}_2 \subseteq \ab$.  }
\end{proof}

To decide the satisfiability of $\phi$, we can use any solver supporting a
sufficiently rich fragment of first-order logic.  The procedure reveals the
concrete executions, as expressed by \mythm\ref{sat:ac:ghb}.

\subsection{Comparison to \cite{bam07} and \cite{sw10,sw11}}
\label{sec:cf-s-w-comparison}
\Daniel{A comparison can't seriously be this elaborate -- terminology
necessary should be elsewhere, and enable a much more succinct comparison!}
Both~\cite{bam07} and~\cite{sw10,sw11} use an SSA encoding similar to our
$\ssa$ of \mysec\ref{sec:c-to-aes}.  The difference resides in the ordering
constraints.

\cite{bam07} encodes total orders over memory accesses. Thus, in contrast to
our clock variables with less-than constraints, \cite{bam07}~uses a Boolean
variable~$M_{xy}$ per pair $(x,y)$, whose value places $x$ and $y$ in a total
order: either $x$ before $y$, or $y$ before $x$. \myprog\ref{prog:fib} has
$3\cdot \text{\lstinline{N}}$ memory accesses per thread, hence \cite{bam07}'s
encoding has $6 \cdot \text{\lstinline{N}} \cdot (6 \cdot \text{\lstinline{N}}
-1)$ Boolean variables. \cite{bam07} builds additional constraints for the
transitive closure; their number is at least cubic in the number of
variables~$M_{xy}$, leading to $\mathcal{O}(\text{\lstinline{N}}^6)$
constraints. 

We only consider relations per address, except for program order and fence
orderings, and do not build transitive closures. The constraints for $\fr$ and
$\ab$ are cubic in the worst case; all others are quadratic. In
\myprog\ref{prog:fib}, the write serialisation is internal, hence $\fr$ is only
quadratic. Hence our number of constraints is
$\mathcal{O}(\text{\lstinline{N}}^2)$. 

\cite{sw10,sw11} use partial orders like us; they note redundancies in their
constraints in~\cite{sw11} but do not explain them, which we do below.
Basically, \cite{sw10,sw11} quantify over all events regardless of their
address, whereas we mostly build constraints per address.
\myfig\ref{fig:facts} shows that the maximal number of events to a single
address is experimentally much smaller than the total number of events.

Our notations correspond to the ones of~\cite{sw11} as follows (the original
description~\cite{sw10} has different notations).  $\HB(a,b)$ is our clock
constraint $c_{ab}$.  The functions $\loc$ and $\val$ map to ours; $\en(x)$ is
our $\guard(x)$; $\link(r,w)$ denotes that~$r$ reads from~$w$, \ie, our
$s_{wr}$.  \cite{sw11} expresses $\po$, $\rf$, $\fr$, and $\ws$ as follows
(since it is restricted to SC, it gives no encoding of $\ppo_A$, $\grf_A$, and
$\ab_A$).

\cite{sw11} encodes $\po$ as the conjunction of the $c_{a_i a_j}$, with $a_i$
in $\po$ before $a_j$. If the implementation of~\cite{sw11} strictly follows
this definition, it redundantly includes the transitive closure constraints,
which we avoid by building the transitive reduction in
\myalg\ref{alg:ppo-new}.

\cite{sw11} encodes $\rf$ in $\Pi_1 := \forall r.  \exists w. \guard(r)
\implies (\guard(w) \wedge s_{wr})$ and $\Pi_2 := \forall r .  \forall w .
s_{wr} \implies (c_{wr} \wedge \loc(r) = \loc(w) \wedge \val(r) = \val(w))$.
\cite{sw11} forces $\rf$ to be exclusive.  
We explained in \mysec\ref{pord:rf} why this is unnecessary in our case, which
allows us to only build a disjunction over writes (\cf \myalg\ref{alg:rf})
linear in their number.  
\IGNORE{JADE M@J: we probably should, but I think it
will make very little difference. What is really really crucial is that they
have their quantifiers include all events, not just those to the same address
J@MICHAEL: we need to compare these two encodings in the experiments}
\IGNORE{MICHAEL: I really think we should, as I've been asked the following:
does adding a constraint to enforce uniqueness help the solver though?  (I've
built SMT formulas where adding redundant constraints improves solve times.
Though i'd agree that leaning towards smaller formulas is better.)}
 
$\Pi_2$ combines our value and clock constraints, with one major difference:
$\Pi_2$ ranges over all reads and writes, regardless of their address.  Our
$\rf$ (\myalg\ref{alg:rf}) ranges over pairs to the same address, thus reaches
this number only when all reads and writes have the same memory address, which
is unlikely in non-trivial programs.  

Ranging over the same address, as we do, and not all addresses, as in
\cite{sw11}, becomes even more advantageous in $\Pi_3$, encoding~$\fr$: $\Pi_3
:= \forall r .  \forall w .  \forall w'. (\s_{wr} \implies (\guard(w') \wedge
\neg c_{w'w} \wedge \neg c_{rw'} \implies \loc(r) \neq \loc(w'))$.
$\Pi_3$ ranges over all $(r,w,w')$, again independently of their addresses. For
distinct addresses the conjunction holds trivially, but \cite{sw11} builds it
nevertheless.  Our $\fr$ (\cf \myalg\ref{alg:fr}) quantifies only over the
same address, thus spares these trivial constraints.

\cite{sw11} does not encode~$\ws$. The totality of~$\ws$ comes as a side
effect: \cite{sw11} initialises each write with a unique integer, hence writes
are totally ordered by $<$ over integers. This is again regardless of the
addresses, whereas we order writes to the same address only. 

\section{Experimental Results\label{sec:experiments}}
We detail here our experiments, which indicate that our technique is scalable
enough to verify non-trivial, real-world concurrent systems code, including the
worker-synchronisation logic of the relational database PostgreSQL, 
code for socket-handover in the Apache httpd, and the core API of the
Read-Copy-Update mutual exclusion code from Linux~3.2.21. 

We implement our technique within the bounded model checker CBMC~\cite{ckl04},
using a SAT solver as an underlying decision procedure. We see two primary
comparison points to estimate the overhead introduced by the partial order
constraints. First, we pass the benchmarks with a single, fixed interleaving to
sequential CBMC. Our implementation performs comparably to sequential CBMC, as
\myfig\ref{fig:facts} shows (rows ``sequential'' and ``concurrent''). Second, we
compare to ESBMC~\cite{cf11}, which also implements bounded model checking, but
uses interleaving-based techniques.

In \myfig\ref{fig:facts}, we gather facts about all examples: the Fibonacci
example from \cite{b12} with \lstinline{N=5}, $4500$ litmus tests (see
below), the worker synchronisation in PostgreSQL, RCU, and
fdqueue in Apache httpd.  For each we give the number of lines of code (LOC),
the number of distinct memory addresses ``tot.~addr'' (including unused
shared variables), the total number of
shared accesses ``tot.~shared'', the maximal number of accesses to a single
address ``same addr'', the total number of constraints ``all constr'' and the
relation with the most costly encoding, in terms of the number of constraints
generated.  We give the loop unrolling bounds ``unroll'': we write ``none''
when there is no loop, and ``bounded'' when the loops in the program are
natively bounded. 

The total number of shared accesses is on average $13$ times the maximal number
of accesses to a single address.  The most costly constraint is usually the
read-from, or the barriers, which build on read-from. The time needed by our
tool to analyse a program grows with the total number of constraints generated.
ESBMC is $4$ times slower than our tool on Fibonacci, $3050$ times slower on
the litmus tests, times out on PostgreSQL, and cannot parse RCU and Apache.

\begin{figure}
\begin{center}
\scalebox{0.9}{
\begin{tabular}{c|c|c|c|c|c}
            & Fibo. & Litmus & PgSQL & RCU   & Apache \\ \hline
LOC         &  41   & 50.9   & 5412  & 5834       & 28864  \\
unroll & 5     & none & 2    & bounded & 5 \\
tot.~addr &  2           & 11.8   & 6            & 3          & 8      \\
tot.~shared &  45          & 58.7   & 233          & 107        & 88     \\
same addr  &  11          &  3.7   & 72           & 4          & 5      \\
all constr &  308         & 874    & 3762         & 90         & 160    \\
most costly &  $\rf$ (178) & $\ab$ (342)       & $\rf$ (1868) & $\rf$ (33) & $\rf$ (49) \\
sequential   &  0.3\,s      & 0.1\,s       & 4.1\,s       & 0.8\,s     & 1.7\,s \\
concurrent  &  3.3\,s      &  0.2\,s      & 90.0\,s      & 1.0\,s     & 2.8\,s 
\\
ESBMC       &  13.8\,s   & 609.8\,s     &  t/o            & parse err           &  parse err
\end{tabular}}
\end{center}
\caption{Facts about all examples\label{fig:facts}}
\end{figure} 

\begin{figure*}
\scalebox{0.8}{
\begin{tabular}{l|c|c|c|c|c|c|c|c|c}
   &  CBMC        &  CBMC         &  CBMC        &  CheckFence &  CImpact &  ESBMC    &  Poirot        &  SatAbs &  Threader \\ \hline
   &  SC          &  TSO          & Power        &  SC, TSO    &  SC      &  SC       &  SC            &  SC     &  SC        \\
 F &  CE $N=300$  &  CE $N=220$   & CE $N=240$   & conv err    & t/o $N=1$& CE $N=10$ & fails $N\geq1$ & V $N=3$ & t/o $N=1$   \\
 L &  100\%       &  100\%        & 100\%        & 18\%        & 20\%     & 34\%      & 47\%           & 100\%   & 8\%   \\
 P &  V           &  V            & CE           &  conv err   &  aborts  & t/o       &  parse err     & t/o       & n/a        \\
 Pf & V           &  V            & V            &  conv err   &  aborts  & t/o       &  parse err     & t/o      & n/a        \\
 R &  V           &  V            & V            &  conv err   &  aborts  & parse err &  parse err     & ref~err & n/a    \\
 A &  V           &  V            & V            &  conv err   &  aborts  & parse err &  parse err     &  aborts    & n/a   
\end{tabular}}
\caption{Comparison of all tools on all examples \label{fig:all-tools} (time out $30$\,mins)}
\vspace*{-5mm}
\end{figure*}


\paragraph{Other tools}
There are very few tools for verifying concurrent C programs, even on
SC~\cite{dkw08}.  For weak memory, existing techniques are restricted to TSO,
and its siblings PSO and RMO~\cite{bam07,kvy10,kvy11,abp11,aac12,lnp12}.  Not
all of them have been implemented, and only few handle systems code given as C
programs.

Thus, as a further comparison point, we implemented an instrumentation
technique~\cite{instrumentation_paper}, similar to~\cite{abp11}. The technique of \cite{abp11} is restricted
to TSO, and consists in delaying writes, so that the SC executions of the
instrumented code simulate the TSO executions of the original program. Our
instrumentation handles all the models of \mysec\ref{sec:model}. 

We tried~$5$ ANSI-C
model checkers: SatAbs, a verifier based on predicate abstraction~\cite{cks05};
ESBMC; CImpact, a variant of the Impact algorithm~\cite{m06} extended to SC
concurrency; Threader, a thread-modular verifier~\cite{gpr11}; and Poirot, which
implements a context-bounded translation to sequential programs~\cite{lr09}.
These tools cover a broad range of techniques for verifying SC programs.  We
also tried CheckFence~\cite{bam07}.

In \myfig\ref{fig:all-tools}, we compare all tools on all examples: F for
\myprog\ref{prog:fib}, L for the litmus tests, P for PostgreSQL with its bug,
Pf for our fix, R for RCU and A for Apache. For L, P, R and A, the bounds are
as in \myfig\ref{fig:facts}; for Pf we take the one of P. For F we try the
maximal \lstinline{N} that the tool can handle within the time out of $30$\,
mins. For each tool, we specify the model below.  
We write ``t/o'' when there is a timeout. We write ``fail'' when the tool gives
a wrong answer. CheckFence provides a conversion module from C to its internal
representation; we write ``conv err'' when it fails.  We write ``parse err''
when the tool cannot parse the example. 
SatAbs uses a refinement procedure; we write ``ref err'' when it fails.
When a tool verifies an example we write ``V''; when it finds a counterexample
we write ``CE''. 

\paragraph{Fibonacci}
All tools, except for ESBMC, SatAbs and ours, fail to analyse
Fibonacci.  Poirot claims the assertion is violated for any
\lstinline{N}, which is not the case for $1\leq\text{\lstinline{N}}\leq5$.
SatAbs does not reach beyond $\text{\lstinline{N}}=4$.  Our tool handles more
than $\text{\lstinline{N}}=300$, which is $30$ times more loop unrolling than
ESBMC, within the same amount of time.

\paragraph{Litmus tests}
We analyse $4500$ tests exposing weak memory artefacts, \eg, instruction
reordering, store buffering, store atomicity relaxation. These tests are
generated by the \prog{diy} tool~\cite{ams12}, which generates assembly
programs with a final state unreachable on SC, but reachable on a weaker model.
For example, \ltest{iriw} (\myfig\ref{fig:iriw}) can only be reached on RMO (by
reordering the reads) or on Power (\emph{idem}, or because the writes are
non-atomic).

We convert these tests into C code, of $50$ lines on average, involving $2$ to
$4$ threads.  Despite the small size of the tests, they prove challenging to
verify, as \myfig\ref{fig:all-tools} shows: most tools, except Blender, SatAbs
and ours, give wrong results or fail in other ways on a vast majority of tests,
even for SC.  For each tool we give the average percentage of correct
results over all models.  Our tool verifies all tests on all models in
$0.22$\,s on average.

\paragraph{PostgreSQL}
Developers observed that a regression test failed on a PowerPC
machine\footnote{http://archives.postgresql.org/pgsql-hackers/2011-08/msg00330.php},
and later identified the memory model as possible culprit: the processor could
delay a write by a thread until after a token signalling the end of this
thread's work had been set. Our tool confirmed the bug, and proved a patch we
proposed. A detailed description of the problem is
in~\cite{instrumentation_paper}.

\paragraph{RCU}
\emph{Read-Copy-Update} (RCU) is a synchronisation mechanism of the Linux
kernel, introduced in version 2.5.  Writers to a concurrent data structure
prepare a fresh component (\eg, list element), then replace the existing
component by adjusting the pointer variable linking to it.  Clean-up of the old
component is delayed until there is no process reading.

%
Thus readers can rely on very lightweight (and thus fast) lock-free
synchronisation only.  The protection of reads against concurrent writes is
fence-free on x86, and uses only a light-weight fence ({\tt lwsync}) on Power.
We verify the original implementation of the 3.2.21 kernel for x86 ($5824$
lines) and Power ($5834$ lines) in less than $1$\,s, using a harness that
asserts that the reader will not obtain an inconsistent version of the
component. On Power, removing the {\tt lwsync} makes the assertion fail.
 
\paragraph{Apache}
The Apache httpd is the most widely used HTTP server software. It supports
a broad range of concurrency APIs distributing incoming requests to a pool
of workers. 

%
%
The fdqueue module ($28864$ lines) is the central part of this mechanism, which
implements the hand-over of a socket together with a memory pool to an idle
worker.  The implementation uses a central, shared queue for this purpose.
Shared access is primarily synchronised by means of an integer keeping track of
the number of idle workers, which is updated via architecture-dependent
compare-and-swap and atomic decrement operations.  Hand-over of the socket and
the pool and wake-up of the idle thread is then coordinated by means of a
conventional, heavy-weight mutex and a signal.  We verify that hand-over
guarantees consistency of the payload data passed to the worker in $2.45$\,s on
x86 and $2.8$\,s on Power.  
\section{Conclusion}

Our experiments demonstrate that weakness is a virtue for programs with bounded
loops. Our proofs suggest that this contention is not limited to bounded loops,
but impracticable as is, since it involves infinite structures. Thus
we believe that this work opens up new possibilities for over-approximation for
programs with unbounded loops, which we hope to investigate in the future.

\bigskip

\emph{Acknowledgements}
We would like to thank Lihao Liang and Alex Horn for their detailed comments on
earlier versions of this paper.

\bibliographystyle{splncs03}
\bibliography{papers/aes}
\end{document}